\documentclass[10pt,final,journal]{IEEEtran}
\usepackage{geometry}   

\usepackage[utf8]{inputenc}
\usepackage{bm}
\usepackage{amsmath}
\usepackage{amsthm}
\usepackage{amsfonts}
\usepackage{mathrsfs}
\usepackage{pifont}
\usepackage{amssymb}
\usepackage{verbatim}
\usepackage{upgreek}
\usepackage{color}
\usepackage[final]{graphicx}
\usepackage{caption} 
\usepackage{subfigure}

\usepackage{cuted}
\usepackage{afterpage} 
 
\usepackage{enumerate}
\usepackage{epsfig}
\usepackage{wrapfig}
\usepackage{epstopdf}
\usepackage{cite}
\usepackage{setspace}
\usepackage{makecell}

\usepackage{xcolor}
\usepackage{tcolorbox}

\usepackage{hyperref}
\hypersetup{hypertex=true, %
	colorlinks=true,
	linkcolor=red,
	anchorcolor=blue,
	citecolor=green}
\allowdisplaybreaks[2]
\usepackage{xr} 

\usepackage{pifont}

\usepackage{upgreek}

\usepackage{booktabs}
\usepackage{multirow}
\usepackage{tabularx} 

\usepackage{tabularx} 
\makeatletter
\newcommand{\multiline}[1]{%
	\begin{tabularx}{\dimexpr\linewidth-\ALG@thistlm}[t]{@{}X@{}}
		#1
	\end{tabularx}
}

\usepackage{cuted}

\usepackage[]{algpseudocode} 
\usepackage{algorithmicx,algorithm}
\algrenewcommand\alglinenumber[1]{ #1:}

\renewcommand{\algorithmicrequire}{{\tt \bf{Input:}}}
\renewcommand{\algorithmicensure}{{\tt \bf{Output:}}}

\newtheoremstyle{note}
{3pt}
{3pt}
{\upshape}
{}
{\bfseries}
{:}
{.5em}
{}
\theoremstyle{note}

\newtheorem{lemma}{Lemma}
\newtheorem{proposition}{Proposition}
\theoremstyle{plain}

\newtheoremstyle{noparentheses}
{\topsep}   
{\topsep}   
{\upshape}  
{0pt}       
{\bfseries} 
{.}         
{5pt plus 1pt minus 1pt} 
{\thmname{#1} \thmnumber{#2} \thmnote{#3}}  
\theoremstyle{noparentheses}

\newtheorem{theorem*}{Theorem}
\newtheorem{lemma*}{Lemma}
\newtheorem{definition*}{Definition}
\newtheorem{corollary*}{Corollary}
\newtheorem{proposition*}{Proposition}
\newtheorem{property*}{Property}

\newcommand{\Tr}{\mbox{\rm Tr}\, }

\newcommand{\Diag}{\mbox{\rm Diag}\, }

\newcommand{\ex}{{\mathbb E}}

\newcommand{\subto}{{\text {s.t.}} }
\newcommand*{\trans}{{\mathsf{T}}}

\newcommand*{\ctrans}{{\mathsf{H}}}

\newcommand{\beq}{\begin{equation}}
	\newcommand{\eeq}{\end{equation}}
\newcommand{\bea}{\begin{array}{*{20}{c}}}
	\newcommand{\eea}{\end{array}}
\newcommand{\disps}{\displaystyle}

\newcommand{\sss}{\scriptscriptstyle}

\newcommand{\upsmile}{\buildrel{\lower3pt\hbox{$\scriptscriptstyle\smile$}} 
\over}

\DeclareMathAlphabet\mathbfcal{OMS}{cmsy}{b}{n}

\def\LRT#1#2{\!
	\raisebox{.2ex}{$
		{{\scriptstyle\;#1}\atop{\displaystyle\gtrless}}
		\atop
		{\raisebox{-1.25ex}{$\scriptstyle\;#2$}}
		$}
	\!}

\newcommand{\bzero}{\bm 0}

\newcommand{\bA}{\bm A}
\newcommand{\ba}{\bm a}
\newcommand{\bB}{\bm B}
\newcommand{\bb}{\bm b}

\newcommand{\bc}{\bm c}
\newcommand{\bD}{\bm D}
\newcommand{\bd}{\bm d}

\newcommand{\bF}{\bm F}

\newcommand{\bg}{\bm g}
\newcommand{\bH}{\bm H}
\newcommand{\bh}{\bm h}
\newcommand{\bI}{\bm I}

\newcommand{\bL}{\bm L}

\newcommand{\bJ}{\bm J}

\newcommand{\bM}{\bm M}

\newcommand{\bn}{\bm n}

\newcommand{\bp}{\bm p}

\newcommand{\bR}{\bm R}
\newcommand{\br}{\bm r}
\newcommand{\bS}{\bm S}
\newcommand{\bs}{\bm s}
\newcommand{\bT}{\bm T}

\newcommand{\bU}{\bm U}
\newcommand{\bu}{\bm u}

\newcommand{\bv}{\bm v}

\newcommand{\bx}{\bm x}

\newcommand{\bz}{\bm z}

\newcommand{\rmd}{{\rm d}}

\newcommand{\rmp}{{\rm p}}

\newcommand{\rmr}{{\rm r}}

\newcommand{\rms}{{\rm s}}

\newcommand{\rmt}{{\rm t}}

\newcommand{\rmv}{{\rm v}}

\newcommand{\calN}{{\cal N}}

\newcommand{\calP}{{\cal P}}

\newcommand{\bbC}{{\mathbb C}}
\newcommand{\bbR}{{\mathbb R}}

\newcommand{\bbH}{{\mathbb H}}
\newcommand{\bbS}{{\mathbb S}}

\newcommand{\bSigma}{\mbox{\bm{{$\mit \Sigma$}}}}

\RequirePackage[normalem]{ulem} 
\RequirePackage{color}\definecolor{RED}{rgb}{1,0,0}\definecolor{BLUE}{rgb}{0,0,1} 

\usepackage{xr}
\begin{document}
\title{Radar Network Waveform Design \\for Target Tracking\\
	
	\thanks{The work of Augusto~Aubry and Antonio~De~Maio was supported in part by the European Union under the Italian National Recovery and Resilience Plan	(NRRP) of NextGenerationEU, partnership on “Telecommunications of the Future” (PE00000001 - Program “RESTART”). The work of Tao~Fan, Xianxiang~Yu, and Guolong~Cui was supported in part by the Postdoctoral Innovation Talent Support Program under Grant BX20250410, in part by the National Natural Science Foundation of China under Grants U24B20188 and 62271126. (\emph{Corresponding author: Antonio~De~Maio}.)}
	
	\thanks{Tao~Fan, Xianxiang~Yu, and Guolong~Cui are with the School of Information and Communication Engineering, University of Electronic Science and Technology of China, Chengdu 611731, China. (e-mail: thaumielfan@gmail.com; \{xianxiangyu,cuiguolong\}@uestc.edu.cn).}	
	\thanks{Luca~Pallotta is with the Department of Engineering (DiING), University of Basilicata, I-85100 Potenza, Italy (e-mail:luca.pallotta@unibas.it).}
	\thanks{Augusto~Aubry and Antonio~De~Maio are with the Department of Electrical and Information Technology Engineering, Universit\`a degli Studi di Napoli “Federico II”, I-80125 Napoli, Italy (e-mail: \{augusto.aubry; ademaio\}@unina.it).}
}
	
\author{\IEEEauthorblockN{Tao~Fan, Augusto~Aubry, \emph{Senior Member, IEEE}, Antonio~De~Maio, \emph{Fellow, IEEE}, \\ Luca~Pallotta, \emph{Senior Member, IEEE}, Xianxiang~Yu, and Guolong~Cui, \emph{Senior Member, IEEE}}
}

\maketitle
\IEEEpeerreviewmaketitle

\begin{abstract}
	\boldmath
	This paper addresses the synthesis of slow-time coded waveforms for single target tracking in a radar network operating under colored Gaussian interference. Based on the Posterior Cram\'er Rao Lower Bound (PCRLB), which characterizes the theoretically optimal accuracy of target state estimation, the problem at each tracking frame is formulated as the minimization of the trace of the PCRLB, together with power budget requirements and a similarity constraint to account for transmitter limitations and appropriate waveform features. To tackle this challenging optimization problem, an approximation solution technique is proposed, aimed at better tracking accuracy than the reference code. The resulting approximated problems, endowed with more tractable objective functions through Taylor-series expansion, are solved using a customized block Majorization-Minimization (block-MM) algorithm. The convergence properties of the  developed procedure are thoroughly analyzed. Numerical results illustrate the accuracy improvements in the target state estimation process, and robust tracking performance under uncertain target state conditions achieved by the proposed technique.
\end{abstract}

\begin{IEEEkeywords}
	Code design, majorization-minimization, radar network, posterior Cram\'er Rao lower bound, target tracking.
\end{IEEEkeywords}

\section{Introduction}
In active sensing systems, the transmitted waveform plays a significant role in determining both detection performance and measurement accuracy\cite{richards2010principles}. Recently, the advancements in flexible digital waveform generation technologies and high-speed signal processing equipments, have empowered radar with the capability to dynamically customize waveforms matching the target and environment characteristics. In this respect, considerable efforts, leveraging on the cognitive situational awareness provided by the tracker regarding the predicted target state and by some environment awareness, have been dedicated to waveform selection or waveform design to maintain high-quality trajectory reconstruction\cite{guerci2010cognitive, bell2015cognitive, farina2017impact,cui2020radar}. 

Two main research lines have emerged. The former is focused on optimizing detection probability and measurement estimation accuracy\cite{bell1993information,li2006signal,de2008code,de2009design,aubry2013knowledge,soltanalian2013joint,tang2021constrained,fan2024joint,li2008range}. The latter seeks to achieve an accurate estimate of the target state at the next time step by selecting the optimal waveform from a finite library of pre-configured options or by designing bespoke transmitted signals\cite{kershaw1994optimal,kershaw1997waveform,sira2009waveform,demaio2017book_ch7}. In particular, a dynamic waveform synthesis was first explored in \cite{kershaw1994optimal}, where the optimal signal parameters were determined for tracking one-Dimensional (1D) target motion using a linear model for the observations under the assumption of ideal detection performance ($P_\rmd=1$ and $P_{\rm fa}=0$) and white Gaussian noise. It was shown that the Fisher Information Matrix (FIM) for target state estimation could be extracted directly from the peak curvature of the waveform ambiguity function. Accordingly, the optimal parameters for three types of fast-time waveforms were chosen using two criteria based on the mean-square tracking error and validation gate volume, respectively. In \cite{kershaw1997waveform}, the above method was extended to cluttered scenarios with detection probability less than one and without false alarm, using the predicted mean-square tracking error as the design criterion. In \cite{rago1998detection} and \cite{niu2002tracking}, the tracking performances of different sub-pulses combinations, involving constant and linear Frequency-Modulated (FM) sub-pulses with positive or negative sweep rate, were compared using the expected value of the steady-state estimation error to guide the transmitted signal structure. The concept of adaptive waveform selection was generalized to Interacting Multiple Model (IMM) trackers with one- and two-step look ahead in \cite{suvorova2005multi} and \cite{suvorova2006waveform}. With the goal of minimizing the cumulative probability of track loss and the target state covariance, a joint selection of detection threshold and waveform for single target tracking was pursued in \cite{hong2005optimization}. For two-Dimensional (2D) maneuvering motion tracking, \cite{demaio2017book_ch7} employed the weighted Mean Square Error (MSE) of the target state estimate as the criterion for selecting the optimal waveform from an FM pulse library. Within the framework of sequential Bayesian inference for tracking, the Posterior Cram\'er-Rao Lower Bound (PCRLB) was used in \cite{hurtado2008adaptive} for pulse selection in a polarimetric radar. Additionally, considering multiple active sensors tracking a target, \cite{sira2006waveform} and \cite{sira2007dynamic} configured linear and nonlinear FM pulses across two sensors to minimize the predicted mean square tracking error. Using the same design criterion, \cite{nguyen2015adaptive} investigated the adaptive scheduling of linear FM Gaussian pulses in the radar network composed of a single transmitter and multiple receivers. Furthermore, joint waveform selection and system resource allocation for single or multiple target tracking via a distributed radar network was explored in \cite{benavoli2019joint, yi2020resource,shi2021joint,yan2022radar}. Recently, several learning-based method for waveform selection have also attracted attention \cite{thornton2022universal,zhu2023cognitive}. These approaches leverage data-driven techniques to adapt waveform selection dynamically, potentially improving tracking performance in complex environments. However, the effectiveness of all the aforementioned methods remains closely tied to the structure of the pre-configured waveforms\cite{niu2002tracking} as well as to the diversity and size of the available waveform library.

Noticeable, only few studies have focused on the problem of waveform synthesis to enhance the accuracy of target state estimation for tracking. In this respect, in \cite{sen2010ofdm}, considering a wideband Orthogonal Frequency Division Multiplexing (OFDM) signaling scheme, the complex weights applied at the transmitter in a co-located Multiple-Input Multiple-Output (MIMO) radar were built in order to maximize the mutual information between the target state and measurement vectors. The work in \cite{huleihel2013optimal} devised the transmitted signal matrix for MIMO radar to estimate target angle and complex amplitude by minimizing the Bayesian Cram\'er-Rao Bound (BCRB) or the Reuven-Messer Bound (RMB). However, these studies focused exclusively on waveform design for single radar stations, leaving waveform optimization for target tracking in distributed radar infrastructures an open problem.

To fill the aforementioned gap, this paper focuses on synthesizing slow-time coded waveforms for a radar network tracking a point-like target in a colored Gaussian interference environment. Specifically, target motion is assumed to follow a linear Constant Velocity (CV) model in a 2D scenario. Each radar node is equipped with a uniform linear array, and emits a slow-time coded (both in amplitude and in phase) pulse train of linear chirps. In this context, capitalizing on the one-step-ahead predicted target state, the theoretical detection probability and the measurement noise covariance matrix of the time-delay, Doppler shift and arrival angle observations at each radar nodes are first derived, showing that they are functions of the slow-time codes. Then, the trace of the PCRLB (evaluated under the assumption of negligible process noise \cite{tichavsky1998posterior, hernandez2004comparison}) is used as waveform optimization criterion. Moreover, to meet basic radar and target tracking requirements, a power budget limitation, along with a similarity constraint that controls relevant waveform characteristics, are imposed on the probing waveforms at each radar node. To address the challenging non-convex nature of the problem, an approximate solution technique is proposed. The approach involves reformulating the trace of the PCRLB into a more tractable expression to enable radars degrees of freedom  optimization. Specifically, a second-order Taylor expansion is employed to achieve this approximation. The resulting problem is then solved using a customized block Majorization-Minimization (block-MM) algorithm with ensured convergence. At the analysis stage, some numeral results are provided to assess the performance of the new waveform resource allocation strategy for radar networks in terms of convergence and target state estimation accuracy.

Summarizing, the main technical contributions of this paper are: 
\begin{enumerate}[a)]
	\item The computation of the PCRLB for target state estimation in a radar network emitting slow-time coded waveforms, and the formulation of the radar codes design problem aimed at minimizing the trace of the PCRLB under some practical constraints.
	\item The development of a solution technique to optimize tracking accuracy, leveraging the block-MM framework to monotonically minimize the resulting objective.
	\item The analysis of several case studies to demonstrate that the devised radar codes would enhance tracking accuracy compared to existing well-known waveforms while maintaining robust tracking performance under uncertain target state conditions.
\end{enumerate}

The rest of the paper is organized as follows. Section \ref{SecII} introduces the signal model for target state estimation and formulates the design problem of the waveforms emitted by the radar nodes of the network. In Section \ref{SecIII}, an approximation of the objective function is introduced based on a Taylor expansion and  solution technique under the umbrella of the block-MM is developed for synthesizing the radar codes. Section \ref{SecIV} presents some case studies to evaluate the performance of the devised methodology. Finally, in Section \ref{SecV}, conclusions and some possible future research lines are provided. 

\emph{Notation:} Throughout the paper, the following notations have been adopted. $\bbR^N$, $\bbC^N$, $\bbR^{N\times M}$, $\bbC^{N\times M}$, $\bbS_{++}^{N}$, $\bbH_{++}^{N}$ are, respectively, the set of $N$-dimensional vectors of real numbers, complex numbers,  ${N\!\times \!M}$ real matrices, ${N\!\times\! M}$ complex matrices, ${N\!\times\! N}$ symmetric positive matrices, and ${N\!\times\! N}$ Hermitian positive matrices. Scalars, vectors, and matrices are denoted by standard lowercase letter $a$, lower case boldface letter $\ba$, and upper case boldface letter $\bA$, respectively. $\odot$ and $\otimes$ denote, respectively, the Hadamard product and the Kronecker product. For any complex number $x$, $\Re\{x\}$, $\Im\{x\}$, and $|x|$ indicate, respectively, the real part, imaginary part, and modulus. Letter $\jmath$ represents the imaginary unit (i.e., $\jmath$ = $\sqrt{-1}$). The symbols $(\cdot)^\trans$, $(\cdot)^{*}$, $(\cdot)^\ctrans$, $(\cdot)^{-1}$,  $\Tr(\cdot)$, $\|\cdot\|$, $\lambda_{\min}(\cdot)$, and $\lambda_{\max}(\cdot)$ denote the transpose, complex conjugate, conjugate transpose, inverse, trace, the Euclidean norm of a vector, the smallest and the largest eigenvalue of a matrix, respectively. Given two sets $A$ and $B$, $A \setminus B$ denotes the set of all
elements in $A$ that are not in $B$. Finally, $\frac{\partial f(x)}{\partial x}\!$ is the derivative of $f(x)$ with respect to $x$. 

\section{System Model \& Problem Formulation}\label{SecII}
Consider a 2D radar network comprising $N$ widely separated radar nodes. Each node has transceiver capabilities and is equipped with a uniform linear array of $N_\rmr $ elements. The location of the $n$-th radar node is denoted as $\bp_{n}=[x_{}^n,y_{}^n]^\trans\in\bbR^2$ for $n\in\calN=\{1,2,\cdots,N\}$. In the following, it is assumed that radar nodes transmit frequency orthogonal signals and each radar is tuned to its own transmitted frequency. Otherwise stated, each sensor perceives the echoes from its own probing signal.

The sensing scenario of interest subsumes the radar network tracking a moving point-like target within its surveillance region. Let $T$ denote the update tracking interval between the two contiguous frames. The kinematic target state at time $kT$ is specified by its location $[x_k,y_k]^\trans\in\bbR^2$ and velocity $[\dot x_k,\dot y_k]^\trans\in\bbR^2$. A notional illustration of the sensing configuration at the discrete time (frame) $k$ is depicted in Fig. \ref{f1}.

In the following subsections, the target characterization in terms of kinematics, radar echo detection and measurement modelling is first established. Subsequently, the target state estimation problem is formulated for single-target tracking in the radar network.

\subsection{Target Dynamic Model}
Let $\bx_k=[x_k, \dot x_k, y_k, \dot y_k]^\trans\in\bbR^4$ be the target state vector at frame $k$. Suppose that the target motion follows a linear CV model\cite{li2003survey}, i.e., 
\beq
	\bx_{k+1} = \bF\bx_k + \bu_k,
\eeq
where the transition matrix $\bF$ is given by 
\beq
\bF=\bI_2 \otimes \disps \left[\!\!\bea 
			1 & T \\
			0 & 1
			\eea\!\!\right],
\eeq
with $\bI_2$ being the $2\times 2$ identity matrix. The term $\bu_k$ is the process noise associated with the target motion, which obeys a zero-mean Gaussian distribution with covariance \cite{li2003survey,bar2011tracking}
\beq
	\bU_k = \delta \bI_2 \otimes  \left[\!\!\bea 
	\frac{1}{3}T^3 & \frac{1}{2}T^2 \\[0.5em]
	\frac{1}{2}T^2 & T
	\eea \!\!\right],
\eeq
with $\delta$ being the power level of the process noise. 

\begin{figure}
	\centering
	\includegraphics[width=2.8in]{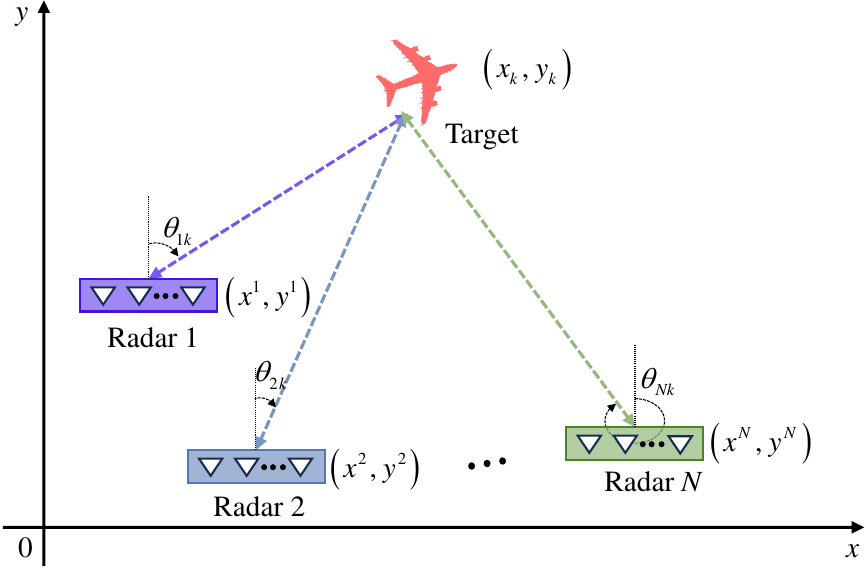}
	\caption{\small Spatial geometry relationship between the radar systems and target at frame $k$. The locations of the $n$-th radar node and the target at frame $k$ are denoted by $\bp_n=[x^n,y^n]^\trans$ and $[x_k,y_k]^\trans$, respectively. $\theta_{nk}$ is the azimuthal angle of the $n$-th node measured clockwise from the normal direction (i.e., the positive y-axis).}\label{f1}
\end{figure}

\subsection{Radar Echo Model}
Suppose that at frame $k$, each radar node emits a burst of $M$ coherent pulses, modulated by a slow-time code. Each pulse is a linear chirp signal $s(t)$ with bandwidth $B$ and pulsewidth $T_{\rm p}$. The baseband transmitted signal of the $n$-th radar node is given by\cite{de2008code}
\beq\label{trans_st}
	{\bar s_{nk}}\left( t \right) = \sum\limits_{m = 0}^{M - 1} {\bc_{nk}( m )s( {t - m{T_{\rmr}}} )},
\eeq
where $T_{\rmr}$ denotes the pulse repetition interval (PRI), and $\bc_{nk} = [\bc_{nk}(0), \bc_{nk}(1), \cdots, \bc_{nk}(M-1)]^\trans\in\bbC^{{M}}$ represents the slow-time code of the $n$-th radar node at frame $k$. Needless to say, $MT_\rmr\leq T$.

The echo signal backscattered from the moving point-like target and received by the $n$-th radar node is downconverted and sampled with a temporal step-size $\Delta t$. Ignoring the intra-pulse Doppler shift, the space-time observation vector of the $n$-th radar can be expressed as\cite{de2009code}
\begin{align}\label{echo}
	\br_{nk} =& \alpha_{nk}(\ba_{\rmt}(f_{\rmd, nk})\!\odot\!\bc_{nk} ) \!\otimes\! \bs(\tau_{nk}) \!\otimes\! \ba_{\rms} (\theta_{nk})\! +\!\bv_{nk}\nonumber\\
	=&\alpha_{nk}\ba_{\rms\rmt, nk }+\bv_{nk},
\end{align}
where $\ba_{\rms\rmt, nk }\!=\!(\ba_{\rmt}(f_{\rmd, nk})\!\odot\!\bc_{nk} ) \!\otimes\! \bs(\tau_{nk}) \!\otimes\! \ba_{\rms} (\theta_{nk})$, and 
\begin{itemize}[]
	\item $\alpha_{nk}$ is the target’s complex amplitude related to the $n$-th radar node at the $k$-th frame accounting for both the target and channel propagation effects.
	\item $\tau_{nk}=2r(\bx_k, \bp_{n})/c$ is the round-trip time with $c$ being the speed of light and $r(\bx_k, \bp_{n})$ being the distance between the target and the $n$-th radar node at frame $k$, that is
	\beq
	r(\bx_k, \bp_{n}) = \sqrt{{( {x_k^{} - x_{}^n} )}^2 + {{( {y_k^{} - y_{}^n} )}^2}}.
	\eeq
	\item  $f_{\rmd, nk}=2v_{\rmd}(\bx_k, \bp_{n})/\lambda_n$ is the Doppler shift of the target on the $n$-th radar node at frame $k$, where $\lambda_n$ is the $n$-th central operating wavelength and $v_{\rmd}(\bx_k, \bp_{n})$ is the radial velocity, i.e.,
	\beq
	v_{\rmd}(\bx_k, \bp_{n}) = -\disps \frac{{( x_k- {x_{}^n}){{\dot x}_k} + ( y_k-{y_{}^n}  ){{\dot y}_k}}}{{\sqrt {{{( {x_k^{} - x_{}^n} )}^2} + {{( {y_k^{} - y_{}^n} )}^2}} }}.
	\eeq
	\item  $\theta_{nk}$ represents the azimuthal angle of the target to the $n$-th radar node at frame $k$ as specified in Fig. \ref{f1}, given by
	\beq
	\theta_{nk} = {\arctan \left[ {\frac{{\left( {x_k^{} - x_{}^n} \right)}}{{\left( {y_k^{} - y_{}^n} \right)}}} \right]}.
	\eeq
	\item $\ba_{\rmt}(f_{\rmd})\in\bbC^{M}$  and $\ba_{\rms} (\theta)\in\bbC^{N_\rmr}$ denote, respectively, the temporal steering vector and spatial steering vector, which can be expressed respectively as
	\begin{align}
		\!\!\!\!\!\ba_{\rmt}(f_{\rmd}) &\!=\! \big[1,e^{\jmath  2\pi f_{\rmd}T_{\rmr}},\cdots,e^{\jmath  2\pi f_{\rmd}(M-1)T_{\rmr}}\big]^\trans,\\
		\!\!\!\!\!\ba_{\rms}(\theta) &\!=\! \big[1,e^{\jmath 2\pi d_n\frac{\sin \theta}{\lambda_n}},\cdots,e^{\jmath 2\pi (N_\rmr-1)d_n\frac{\sin \theta}{\lambda_n}}\big]^\trans,\!
	\end{align}
	with $d_n$ the array inter-element spacing of the $n$-th sensor.
	\item $\bs(\tau_{nk}) \in \bbC^{N_\rmp}$ is the vector collecting the samples of the fast-time pulse $s(t)$ delayed by  $\tau_{nk}$, i.e., $\bs(\tau_{nk}) =[s(\tau_{0, nk}-\tau_{nk}), s(\tau_{0, nk}+\Delta t-\tau_{nk}), \cdots, s(\tau_{0, nk}+(N_\rmp-1)\Delta t-\tau_{nk})]^\trans\in\bbC^{N_{\rmp}}$, where $\tau_{0, nk}\approx\tau_{nk}$ is the first time sample where the echo from the range cell of interest occurs, $N_{\rmp}\Delta t=T_\rmp$, and  $N_{\rmp}$ is the number of samples.
	\item $\bv_{nk}\in\bbC^{MN_\rmp N_\rmr}$ denotes the signal-independent interference vector of the $n$-th radar node at frame $k$, including additive noise, jamming, radio frequency interference, and hot clutter. It is modeled as a zero-mean complex circular Gaussian vector with positive definite covariance matrix $\mathbb{E}[\bv_{nk}\bv_{nk}^\ctrans]=\bSigma_{{{\rmv}, nk} }\in\bbH_{++}^{{MN_\rmp N_\rmr}}$. Assuming statistical separability across slow-time, fast-time, and space dimensions, together with stationarity within each domain, the overall covariance can be approximated as a Kronecker product of the following lower-dimensional covariances \cite{srivastava2008models}: 
	\beq
	\bSigma_{{\rmv,  nk} } = \bSigma_{{\rmt}, nk} \otimes \bI_{N_\rmp} \otimes \bSigma_{\rms, nk},
	\eeq
	where $\bI_{N_\rmp}$ is the covariance matrix of the signal-independent interference in the fast-time domain (assumed white), whereas, $\bSigma_{{\rmt}, nk}\in\bbH_{++}^{{M}}$ and $\bSigma_{{\rms}, nk}\in\bbH_{++}^{{N_\rmr}}$ are, respectively, the covariance matrices in slow-time and space. 
\end{itemize}

\subsection{Detection and Measurement Model}

\subsubsection{Detection Model}
The problem of detecting a target using the signal model in \eqref{echo} can be formulated as the following binary hypothesis test
\beq \label{binhypo}
\bigg\{{\begin{array}{*{20}{l}}
		{H_{0, nk}}\!: \br_{nk} = \bv_{nk} \\
		{H_{1, nk}}\!: \br_{nk} = \alpha_{nk}\ba_{\rms\rmt, nk } + \bv_{nk}
\end{array}}.
\eeq
Assuming that $\alpha_{nk}$ is an unknown deterministic parameter, the Generalized Likelihood Ratio Test (GLRT) detector\footnote{According to the Neyman-Pearson criterion, if the phase of $\alpha_{nk}$ is uniformly distributed over $[0, 2\pi)$ and the interference covariance matrix $\bSigma_{{\rmv}, nk}$ is known, then the GLRT detector is equivalent to the optimum test.} is given by
\beq
\big|\br_{nk}^\ctrans\bSigma_{{\rmv}, nk}^{-1}\ba_{\sss\rms\rmt, nk }\big|^2\LRT{H_{1, nk}}{H_{0, nk}} \eta,
\eeq
where $\eta$ indicates the detection threshold determined by the desired value of the false alarm probability ($P_{\text{fa}}$). Moreover, the detection probability of the $n$-th radar node at frame $k$ can be written as
\beq\label{Pd}
P_{\rmd, nk}(\bc_{nk}) = Q_1\Big(\sqrt{2 \text{SINR}(\bc_{nk}, \tau_{nk}, f_{\rmd, nk}, \theta_{nk})}, \sqrt{2b_0}\Big), 
\eeq
where $Q_1(\cdot, \cdot)$ is the Marcum Q function of order 1, $b_0=-\ln P_{\text{fa}}$, and $\text{SINR}(\bc_{nk}, \tau_{nk}, f_{\rmd, nk}, \theta_{nk})$ is the output SINR, defined in \eqref{SINR1} at the top of the next page. 
\begin{figure*}
\vspace{-3mm}
\beq\label{SINR1}
	\text{SINR}(\bc_{nk}, \tau_{nk}, f_{\rmd, nk}, \theta_{nk})= |\alpha_{nk}|^2\|\bs(\tau_{nk})\|^2 \Big(\ba^\ctrans_{\rms} (\theta_{nk})\bSigma_{{\rms}, nk}^{-1}\ba_{\rms} (\theta_{nk}) (\ba_{\rmt}(f_{\rmd, nk}) \! \odot\! \bc_{nk})^\ctrans\!\bSigma_{{\rmt}, nk}^{-1} (\ba_{\rmt}(f_{\rmd, nk}) \! \odot \! \bc_{nk})\Big)
	\vspace{-2mm}
\eeq
\vspace{-3mm}
\hrulefill
\end{figure*}

Under the assumption of $\tau_{0, nk} \approx \tau_{nk}$, the term $\|\bs(\tau_{nk})\|^2$ in \eqref{SINR1} can be approximated by $\|\bs(\tau_{nk})\|^2=N_\rmp$. Then, the output SINR can be simplified as\footnote{For any finite sampling time, the reported equality holds true only approximately. However, as long as the sampling interval is short enough, the approximation becomes tighter and tighter.}
\beq
\text{SINR}(\bc_{nk}, \tau_{nk}, f_{\rmd, nk}, \theta_{nk}) = \varepsilon_{\alpha, nk} \bc_{nk}^\ctrans\bM_{0, nk}\bc_{nk},
\eeq
where $\bM_{0, nk}=\bSigma_{{\rmt}, nk}^{-1} \odot (\ba_{\rmt}(f_{\rmd, nk})\ba_{\rmt}(f_{\rmd, nk})^\ctrans)\in \bbH^{M}$, and $\varepsilon_{\alpha, nk} = N_\rmp|\alpha_{nk}|^2 \ba^\ctrans_{\rms} (\theta_{nk})\bSigma_{{\rms}, nk}^{-1}\ba_{\rms} (\theta_{nk})$.
\subsubsection{Measurement Model} Assuming that the false alarm probability\footnote{In practice, the false alarm probability of radar systems is typically set between $10^{-4}$ and $10^{-6}$ \cite{richards2010principles}.} $P_{\rm fa}\approx0$, the measurement model, also referred to as the observation model\cite{bar2011tracking}, for range (time-delay), velocity (Doppler shift) and arrival angle is expressed as
\beq\label{measuremodel}
\bz_{nk} = \bh_n({\bx_k}) +\bn_{nk},
\eeq
where $\bh_n(\bx_k) \!=\! [r(\bx_k, \bp_{n}), v_{\rmd}(\bx_k, \bp_{n}), \theta(\bx_k, \bp_{n})]^\trans$ denotes the ground-truth measurement of the $n$-th radar node at frame $k$, and $\bn_{nk}$ is the measurement noise vector with covariance matrix $\bR_{nk}(\bc_{nk})$.

It is worth stressing that the measurement noise covariance matrix $\bR_{nk}(\bc_{nk})$ is dependent on the slow-time codes and can be computed as
\beq
	\bR_{nk}(\bc_{nk}) = \bT_n\textbf{CRLB}_{nk}(\bc_{nk}) \bT_n,
\eeq
where $\bT_n=\Diag([c/2, \lambda_n/2, 1]^\trans)$ is a matrix that allows to convert time-delay and Doppler shift to physical distance and velocity, whereas $\textbf{CRLB}_{nk}(\bc_{nk})$ is the CRLB matrix for the estimation of $\tau_{nk}$, $f_{\rmd, nk}$, and $\theta_{nk}$ at radar node $n$ in frame $k$. According to the space-time observation vector defined in \eqref{echo}, as shown in Appendix \ref{proofCRLB} of the supplementary material, $\textbf{CRLB}_{nk}(\bc_{nk})$ is a diagonal matrix whose diagonal elements are  
	\begin{subequations}\label{CRLB}
	\!\!\!\!\!\!\begin{align}
		\begin{split}\label{CRLBt}
			[\textbf{CRLB}_{nk}(\bc_{nk})]_{11} &= \frac{1}{\varepsilon_{\tau, nk}} \displaystyle \frac{1}{\bc_{nk}^\ctrans \bM_{0, nk} \bc_{nk}},
		\end{split}\\
		\begin{split}\label{CRLBf}
			[\textbf{CRLB}_{nk}(\bc_{nk})]_{22} &= \frac{1}{\varepsilon_{f_\rmd, nk}} \displaystyle \frac{\bc_{nk}^\ctrans \bM_{0, nk} \bc_{nk}}{\phi(\bc_{nk})},
		\end{split}\\
		\begin{split}\label{CRLBtheta}
			[\textbf{CRLB}_{nk}(\bc_{nk})]_{33} &= \frac{1}{\varepsilon_{\theta, nk}} \displaystyle \frac{1}{\bc_{nk}^\ctrans \bM_{0, nk} \bc_{nk}},
		\end{split}
	\end{align}
\end{subequations}
where 
\beq
\phi(\bc_{nk})\! =\! \bc_{nk}^\ctrans \!\bM_{0, nk} \bc_{nk}\bc_{nk}^\ctrans \!\bM_{2, nk} \bc_{nk} \!-\!|\bc_{nk}^\ctrans \bM_{1,nk} \bc_{nk}|^2\!\!\!\nonumber
\eeq
with $\bM_{1, nk} = \Diag(\bb_{\rmt}^*)\bM_{0, nk}\!\in\!\bbC^{M\times M}$, $\bM_{2, nk} = \bM_{0, nk}\odot(\bb_{\rmt}\bb^\ctrans_{\rmt})\!\in\!\bbH^M$, $\bb_{\rmt} = [0, \jmath 2\pi T_{\rmr}, \cdots, \jmath 2\pi  T_{\rmr}(M\!-\!1)]^\trans\!\in\!\bbC^M$.

The definitions of $\varepsilon_{\tau, nk}$, $\varepsilon_{f_\rmd, nk}$, $ \varepsilon_{\theta, nk}$ can be found in equations \eqref{epss_nk} of Appendix \ref{proofCRLB} of the supplementary material.

\subsection{Target State Estimation Problem Formulation}
Suppose the target state $\bx_k$ is estimated by an unbiased estimator $\widehat \bx_k$. The PCRLB for its covariance is defined as the inverse of the Information Matrix (IM)\cite{tichavsky1998posterior}, i.e.,
\beq
	\ex [(\bx_k -\widehat \bx_k)(\bx_k -\widehat \bx_k)^\ctrans] \geq \bJ_k^{-1},
\eeq     
where $\bJ_k$ is the IM at frame $k$. Assuming that the target's process noise is sufficiently weak\footnote{The process noise in the CV model can be interpreted as an acceleration disturbance. This disturbance can be optimized or mitigated using advanced path control techniques\cite{li2003survey,wu2024tunnel}.} to be neglected, the IM at frame $k$ can be calculated via a Riccati-like recursion \cite{tichavsky1998posterior, hernandez2004comparison} as follows
\begin{align}\label{FIM}
	&\bJ_k(\bc_{1k}, \bc_{2k},\cdots, \bc_{Nk}) \nonumber\\
	&= (\bF^{-1})^\trans\bJ_{k-1}(\bc_{1(k-1)}, \bc_{2(k-1)},\cdots, \bc_{N(k-1)})\bF^{-1} \nonumber\\
	&\quad+  \sum\limits_{n\in\calN} P_{\rmd, nk}(\bc_{nk})\bH^\trans_{nk}\bR_{nk}^{-1}(\bc_{nk})\bH_{nk},
\end{align} 
where $\bJ_{k-1}$ is the IM at frame $k-1$, and $\bH_{nk}$ is the Jacobian of $\bh_n(\bx_k)$ given in \eqref{H_nk} at the top of the next page. The initial IM $\bJ_{0}$ required for recursion is typically derived from the prior target state covariance \cite{hernandez2004comparison}. Without loss of generality, $\bJ_{0}$ is initialized as a diagonal matrix  with small positive entries to signify the lack of informative prior knowledge and ensure invertibility.

\begin{figure*}
\beq\label{H_nk}
\bH_{nk} \!=\! \disps\frac{1}{r(\bx_k, \bp_{n})} \!\disps \left[\!\! {\begin{array}{*{20}{c}}
		{x_k^{} - x^n}&0&{y_k^{} - y^n}&0\\[0.5em]
		{-{\dot x}_k} - \disps\frac{\left( {x_k^{} - x^n} \right)v_{\rmd}(\bx_k, \bp_{n})}{r(\bx_k, \bp_{n})}&{x^n - x_k^{}}&-\dot y_k^{} - \disps\frac{\left( {y_k^{} - y^n} \right){v_{\rmd}(\bx_k, \bp_{n})}}{r(\bx_k, \bp_{n})}&{y^n - y_k^{}}\\[0.5em]
		\disps\frac{{y_k^{} - y^n}}{r(\bx_k, \bp_{n})}&0&\disps\frac{{x_k^{} - x^n}}{r(\bx_k, \bp_{n})}&0
\end{array}} \!\!\right]\!\!\!\!\!\!\!
\eeq
\vspace{-3mm}
\hrulefill
\end{figure*}

Notice that the IM at frame $k$ depends on the radar codes transmitted at the previous $k-1$ frames as well as the current one. The functional dependence of \eqref{FIM} on $\bc_{1\bar k}, \bc_{2\bar k},\cdots, \bc_{N\bar k},\bar k=1,\cdots,k$, highlights that a tailored design of these sequences can be exploited to enhance the accuracy of the target state estimation. According to \eqref{FIM}, radar codes can be devised for each tracking frame, provided that the codes from previous frames are fixed. Specifically, given $\bc_{1\bar k}, \bc_{2\bar k},\cdots, \bc_{N\bar k},\bar k=1,\cdots,k-1$, the goal is to optimize the codes for the $k$-th frame. In this context, the trace of the IM inverse (i.e., the trace of the PCRLB) is used as the figure of merit to minimize, that is
\beq
\!\!\!\!\!f_k(\bc_{1k}, \bc_{2k},\cdots, \bc_{Nk}) \!=\! \Tr(\bJ_k^{-1}(\bc_{1k}, \bc_{2k},\cdots, \bc_{Nk})).\!\!
\eeq

Additionally, to regulate the power budget and maintain some desired characteristics of the probing signal, energy and similarity constraints \cite{de2008code,cui2014mimo} are forced on the radar sequence: $\|\bc\|^2 = 1$, $\|\bc-\bc_0\|^2 \leq \zeta$, where $\bc_0$ denotes a reference code with unit energy $\|\bc_0\|^2\!=\!1$, and $0\leq\zeta\leq 2$ is a real parameter ruling the degree of similarity. 

Based on the above discussion, the radar code design problem at frame $k$ can be formulated as 
\beq \label{TrPCRLB0}
\calP^{'}_{k}\left\{\begin{array}{lll}
	\min\limits_{\{\bc_{nk}\}} & f_k(\bc_{1k}, \bc_{2k},\cdots, \bc_{Nk})\\
	\subto & \|\bc_{nk}\|^2 = 1,n\in\calN\\
	& \|\bc-\bc_0\|^2 \leq \zeta,n\in\calN
\end{array},
\right.
\eeq 
which can be equivalently expressed, by utilizing the unit energy restriction, as 
\beq \label{TrPCRLB}
\calP_{k}\left\{\begin{array}{lll}
	\min\limits_{\{\bc_{nk}\}} & f_k(\bc_{1k}, \bc_{2k},\cdots, \bc_{Nk})\\
	\subto & \|\bc_{nk}\|^2 = 1,n\in\calN\\
	& \Re\{\bc_0^\ctrans\bc_{nk}\}\geq \zeta_1,n\in\calN
\end{array},
\right.
\eeq 
where $\zeta_1=1-\zeta/2\geq 0$. It is evident that,  throughout the tracking process, $\calP_k$ needs to be solved at each frame. 

Note that $\calP_{k}$ is a challenging optimization problem due to its nonlinear and non-convex objective function and non-convex constraints. 
\section{Radar Network Code Synthesis for Single Target Tracking} \label{SecIII}
In this section, a solution technique is proposed to handle $\calP_{k}$. In each tracking frame, the goal is to find a solution that achieves a target state estimation accuracy better than that of the reference code $\bc_0$. The approach involves approximating the objective function with a more tractable expression, thereby simplifying the associated optimization problem. This approximated code design problem is then solved to generate candidate radar waveforms. Only if the resulting PCRLB is smaller than the value ensured by the reference codes, the waveforms on the nodes are updated; otherwise, the reference code is retained.

To construct the approximation, the real-valued formulation of $\calP_{k}$ is considered first. Letting $\widetilde \bc_{nk}=[\Re\{\bc_{nk}\}^\trans, \Im\{\bc_{nk}\}^\trans]^\trans$ and $\widetilde \bc_{0}=[\Re\{\bc_{0}\}^\trans,$ $\Im\{\bc_{0}\}^\trans]^\trans$, the quantities $ P_{\rmd, nk}( \bc_{nk})$ and the diagonal entries of $\bR_{nk}(\bc_{nk})$ can be, respectively, expressed as functions of $\widetilde \bc_{nk}$, i.e.,
\begin{align}\label{P_d1}
	\!\!\!\!\widetilde P_{\rmd, nk}(\widetilde\bc_{nk}) \!= \!Q_1\Big(\!\!\sqrt{2 \widetilde{\text{SINR}}(\widetilde\bc_{nk}, \tau_{nk}, f_{\rmd, nk}, \theta_{nk})}, \sqrt{2b_0}\Big),
\end{align}
and
\begin{subequations}\label{Rc1}
	\begin{align}
		&[\widetilde\bR_{nk}(\widetilde \bc_{nk})]_{11} = \frac{c^2}{4\varepsilon_{\tau, nk}} \displaystyle \frac{1}{\widetilde\bc_{nk}^\trans \widetilde\bM_{0, nk} \widetilde\bc_{nk}},\\
		&[\widetilde\bR_{nk}(\widetilde\bc_{nk})]_{22}  = \frac{\lambda_n^2}{4\varepsilon_{f_\rmd, nk}} \displaystyle \frac{\widetilde\bc_{nk}^\trans \widetilde\bM_{0, nk} \widetilde\bc_{nk}}{\widetilde\phi(\widetilde\bc_{nk})},\\
		&[\widetilde\bR_{nk}(\widetilde\bc_{nk})]_{33} = \frac{1}{\varepsilon_{\theta, nk}} \displaystyle \frac{1}{\widetilde\bc_{nk}^\trans \widetilde\bM_{0, nk} \widetilde\bc_{nk}},
	\end{align}
\end{subequations}
where
\begin{align}
	&\widetilde{\text{SINR}}(\widetilde\bc_{nk}, \tau_{nk}, f_{\rmd, nk}, \theta_{nk})=\varepsilon_{\alpha, nk} \widetilde\bc_{nk}^\trans\widetilde\bM_{0, nk}\widetilde\bc_{nk},\\
	&\widetilde\phi(\widetilde\bc_{nk})=\widetilde\bc_{nk}^\trans \!\widetilde\bM_{0, nk} \widetilde\bc_{nk}\widetilde\bc_{nk}^\trans \!\widetilde\bM_{2, nk} \widetilde\bc_{nk} \nonumber\\
	&\quad\quad\quad\quad-(\widetilde\bc_{nk}^\trans \widetilde\bM_{1,nk} \widetilde\bc_{nk})^2-(\widetilde\bc_{nk}^\trans \widehat\bM_{1,nk} \widetilde\bc_{nk})^2,
\end{align}
and 
\begin{subequations}
	\begin{align}
		\widetilde \bM_{0, nk} &= \left[\begin{array}{*{20}{c}}\Re\{\bM_{0, nk}\}&-\Im\{\bM_{0, nk}\}\\\Im\{\bM_{0, nk}\}&\Re\{\bM_{0, nk}\}\end{array}\right],\\
		\widetilde \bM_{1, nk} &= \left[\begin{array}{*{20}{c}}\Re\{\bM_{1, nk}\}&-\Im\{\bM_{1, nk}\}\\\Im\{\bM_{1, nk}\}&\Re\{\bM_{1, nk}\}\end{array}\right],\\
		 \widehat \bM_{1, nk} &= \left[\begin{array}{*{20}{c}}\Im\{\bM_{1, nk}\}&\Re\{\bM_{1, nk}\}\\-\Re\{\bM_{1, nk}\}&\Im\{\bM_{1, nk}\}\end{array}\right],\\
		\widetilde \bM_{2, nk} &= \left[\begin{array}{*{20}{c}}\Re\{\bM_{2, nk}\}&-\Im\{\bM_{2, nk}\}\\\Im\{\bM_{2, nk}\}&\Re\{\bM_{2, nk}\}\end{array}\right].
	\end{align}
\end{subequations}

Then, defining\footnote{Needless to say, $\widetilde P_{\rmd, nk}(\widetilde \bc_{nk})\!>\!0$ during the tracking process.} $\bS_{nk}(\widetilde\bc_{nk})\!\!=\!\!\widetilde\bR_{nk}(\widetilde\bc_{nk})/\widetilde P_{\rmd, nk}(\widetilde \bc_{nk})$, the functional dependency of the matrix $\bJ_k(\bc_{1k}, \bc_{2k},\cdots, \bc_{Nk})$ on the radar codes can be specified via $\widetilde \bc_{nk}$, namely,
 \begin{align}\label{FIM1}
	&\widetilde \bJ_k(\widetilde\bc_{1k}, \widetilde\bc_{2k},\cdots, \widetilde\bc_{Nk}) \nonumber\\
	&= (\bF^{-1})^\trans\widetilde\bJ_{k-1}\bF^{-1} +  \sum\limits_{n\in\calN} \bH^\trans_{nk}\bS_{nk}^{-1}(\widetilde \bc_{nk})\bH_{nk}.
\end{align} 
where $\widetilde\bJ_{k-1}$ is used to synthetically denote $\widetilde\bJ_{k-1}(\widetilde\bc_{1(k-1)}, \widetilde\bc_{2(k-1)},\cdots, \widetilde\bc_{N(k-1)})$. 

Therefore, $\calP_{k}$ can be equivalently expressed as
\beq \label{r_TrPCRLB}
\widetilde\calP_{k}\left\{\begin{array}{lll}
	\min\limits_{\{\widetilde \bc_{nk}\}} & \widetilde f_k(\widetilde\bc_{1k}, \widetilde\bc_{2k},\cdots, \widetilde\bc_{Nk})\\
	\subto & \|\widetilde \bc_{nk}\|^2 = 1,n\in\cal\calN\\
	& \widetilde \bc_0^\trans\widetilde\bc_{nk}\geq \zeta_1,n\in\cal\calN
\end{array},
\right.
\eeq 
where 
\beq\label{tildef}
\!\!\!\!\!\!\widetilde f_k(\widetilde\bc_{1k}, \widetilde\bc_{2k},\cdots, \widetilde\bc_{Nk})\!=\! \Tr\!(\widetilde\bJ_k^{-1}(\widetilde\bc_{1k}, \widetilde\bc_{2k},\cdots, \widetilde\bc_{Nk})).\!\!\!\!\!\!
\eeq 

Notably, only the diagonal matrix $\bS_{nk}(\widetilde \bc_{nk})$, which appears in the objective function of $\widetilde \calP_k$, depends on the radar codes $\widetilde\bc_{nk}$. The positivity of its diagonal elements in \eqref{P_d1} and \eqref{Rc1}, coupled with their dependence on the radar codes, introduces nonlinearities that complicate the optimization. To circumvent this drawback, it is reasonable to approximate these elements using well-established methodologies, such as Taylor series expansions. Let the resulting approximation be denoted by $\widehat \bS_{nk}(\widetilde\bc_{nk})$, the approximated objective function $\widehat f_k(\widetilde\bc_{1k}, \widetilde\bc_{2k},\cdots, \widetilde\bc_{Nk})$ is obtained replacing in \eqref{FIM1} $\bS_{nk}(\widetilde\bc_{nk})$ with $\widehat \bS_{nk}(\widetilde\bc_{nk})$. It is important to emphasize that the ideal cost function $\widetilde f_k(\widetilde\bc_{1k}, \widetilde\bc_{2k},\cdots, \widetilde\bc_{Nk})$ is the sum of terms strictly greater than zero. Accordingly, the adopted approximation must preserve this property, possibly by limiting the search space to prevent an ill-posed cost function. In this respect, a modified approximated objective function with a correction term is introduced as
\beq\label{limitedspace}
	\widehat f_k(\widetilde\bc_{1k}, \widetilde\bc_{2k},\cdots, \widetilde\bc_{Nk}) + \mathbb{I}_{+}\big(\widetilde\bc_{1k}, \widetilde\bc_{2k},\cdots, \widetilde\bc_{Nk}\big),
\eeq
where 
\beq
\mathbb{I}_{+}(\widetilde\bc_{1k}, \widetilde\bc_{2k},\cdots, \widetilde\bc_{Nk}) \!=\! \left\{\!\!\!\begin{array}{lll}
	0, & \min\limits_{n\in\calN,\atop l=1,2,3} [\widehat \bS_{nk}(\widetilde\bc_{nk})]_{ll} \geq \varepsilon, \\
	+\infty, &\text{otherwise}
\end{array}
\right.\nonumber
\eeq
with $\varepsilon$ being a quite small positive constant. Finally, the radar code design problem with the modified approximated objective function is considered.

Following these guidelines, the solution technique developed to solve the approximated version of $\widetilde \calP_k$ is reported in Algorithm \ref{algo1}. Inspection of the steps reveals that the two crucial tasks are the construction of the approximated objective function $\widehat f_k(\widetilde\bc_{1k},\widetilde\bc_{2k}, \cdots\!,\! \widetilde\bc_{Nk})$ and the solution of the corresponding approximated problem $\widehat \calP_{k}$. The following subsections focus on these two aspects.
\begin{algorithm}
	\setstretch{1.1}
	\caption{The approximate solution technique to solve $\widetilde \calP_{k}$.}\label{algo1} 
	\algorithmicrequire{ \multiline{$\widetilde \bc_{0}$, $\zeta$, $\bF$, $\widetilde\bJ_{k-1}$, and $\{\bH_{nk}\}_{n\in\calN}$.}}
	\algorithmicensure{ \multiline{Optimized radar codes $\widetilde\bc_{1k},\widetilde\bc_{2k},\cdots\!,\! \widetilde\bc_{Nk}$ and the\\ achieved IM $\widetilde\bJ_{k}(\widetilde\bc_{1k}, \widetilde\bc_{2k},\cdots, \widetilde\bc_{Nk})$.}} 
	\begin{algorithmic}[1]
		\State \multiline{ Determine an approximation $\!\widehat \bS_{nk}\!(\widetilde\bc_{nk})\!$ to $\!\bS_{nk}\!(\widetilde\bc_{nk})$;}
		\State \multiline{Construct the corresponding approximated objective function $\widehat f_k(\widetilde\bc_{1k},\widetilde\bc_{2k},\cdots\!,\! \widetilde\bc_{Nk})$ replacing $\bS_{nk}(\widetilde\bc_{nk})$ in 
			\eqref{FIM1} with $\widehat \bS_{nk}(\widetilde\bc_{nk})$;}
		\State Solve the following problem to get the candidates $\widetilde \bc_{nk}^\star, n\in\calN$,
		\beq
		\!\!\!\!\!\!\widehat \calP_k\left\{\!\!\!\!\begin{array}{lll}
			\min\limits_{\{\widetilde\bc_{nk}\}} & \widehat f_k(\widetilde\bc_{1k},\widetilde\bc_{2k},\cdots\!,\! \widetilde\bc_{Nk})\!+\!\mathbb{I}_{+}\big(\widetilde\bc_{1k}, \widetilde\bc_{2k},\cdots\!,\! \widetilde\bc_{Nk}\big)\\
			\subto & \|\widetilde\bc_{nk}\|^2 = 1,n\in\calN\\
			& \widetilde \bc_0^\trans\widetilde\bc_{nk}\geq \zeta_1,n\in\calN;
		\end{array}
		\right.\nonumber
		\eeq
		\If{$\widetilde f_k(\widetilde\bc_{1k}^\star,\widetilde\bc_{2k}^\star,\cdots\!,\! \widetilde\bc_{Nk}^\star)< \widetilde f_k(\widetilde\bc_{0},\widetilde\bc_{0},\cdots\!,\! \widetilde\bc_{0})$}
		\State Let $\widetilde\bc_{nk}\!=\!\widetilde\bc_{nk}^\star, n\in\calN$;
		\Else 
		\State Let $\widetilde\bc_{nk}\!=\!\widetilde\bc_{0}, n\in\calN$;
		\EndIf
		\State Compute $\widetilde\bJ_{k}(\widetilde\bc_{1k}, \widetilde\bc_{2k},\cdots, \widetilde\bc_{Nk})$ defined in \eqref{FIM1}.
	\end{algorithmic}
\end{algorithm}
\subsection{Approximation of ${\widetilde f_k}(\widetilde {\boldsymbol c}_{1k},\widetilde{\boldsymbol c}_{2k},\cdots\!,\! \widetilde{\boldsymbol c}_{Nk})$} 
The core idea underlying the approximation ${\widetilde f_k}(\widetilde {\boldsymbol c}_{1k},\widetilde{\boldsymbol c}_{2k},\cdots\!,\! \widetilde{\boldsymbol c}_{Nk})$, as already highlighted, is to suitably describe the diagonal elements of $\bS_{nk}(\widetilde\bc_{nk})$ leveraging a second order Taylor expansion. The following lemma provides details in this regard.

\begin{lemma}\label{prop2}
	The second order Taylor expansion of the $l$-th diagonal elements of $\bS_{nk}(\widetilde\bc_{nk})$ in the neighborhood of the reference code $\widetilde \bc_{0}$ is given by 
	\beq\label{hatRc}
		[\widehat \bS_{nk}(\widetilde\bc_{nk})]_{ll} =  \widetilde \bc_{nk}^\trans\bA_{l, nk} \widetilde \bc_{nk} + \ba_{l, nk}^\trans 	\widetilde \bc_{nk} + a_{l, nk},
	\eeq
	where $\bA_{l, nk}$, $\ba_{l, nk}$, and $a_{l, nk}$ are reported in Appendix \ref{proof2ndexpansion} of the supplementary material.
\end{lemma}

By approximating $\bS_{nk}(\widetilde\bc_{nk})$ in $\widetilde f_k(\widetilde\bc_{1k}, \widetilde\bc_{2k},\cdots, \widetilde\bc_{Nk})$ via the expansion formula in \eqref{hatRc}, the resulting function  $\widehat f_k(\widetilde\bc_{1k},\widetilde\bc_{2k},\cdots\!,\! \widetilde\bc_{Nk})$ can be expressed as 
\beq\label{approx_fk}
\!\!\!\!\!\!\widehat f_k(\widetilde\bc_{1k},\widetilde\bc_{2k},\cdots\!,\! \widetilde\bc_{Nk}) \!=\! \Tr(\widehat\bJ_k^{-1}(\widetilde\bc_{1k}, \widetilde\bc_{2k},\cdots, \widetilde\bc_{Nk})),\!\!\!\!\!\!
\eeq
where 
\begin{align}\label{approx_Jk}
	&\widehat \bJ_k(\widetilde\bc_{1k}, \widetilde\bc_{2k},\cdots, \widetilde\bc_{Nk}) \nonumber\\
	&= (\bF^{-1})^\trans\widetilde\bJ_{k-1}\bF^{-1} +  \sum\limits_{n\in\calN} \bH^\trans_{nk}\widehat\bS_{nk}^{-1}(\widetilde \bc_{nk})\bH_{nk}.
\end{align} 
\subsection{Solution to Problem $\widehat \calP_k$}
Before proceeding further, let us focus on an equivalent expression of $\widehat \calP_k$, i.e., 
\beq\label{Pk1}
\bar \calP_k\left\{\!\!\!\!\begin{array}{lll}
	\min\limits_{\{\widetilde\bc_{nk}\}} & \widehat f_k(\widetilde\bc_{1k},\widetilde\bc_{2k},\cdots\!,\! \widetilde\bc_{Nk})\\
	\subto & \|\widetilde\bc_{nk}\|^2 = 1,n\in\calN\\
	& \widetilde \bc_0^\trans\widetilde\bc_{nk}\geq \zeta_1,n\in\cal\calN\\
	&[\widehat \bS_{nk}(\widetilde\bc_{nk})]_{ll}\geq \varepsilon,l=1,2,3,n\in\calN
\end{array}.
\right.
\eeq

Inspired by the block–MM framework\cite{razaviyayn2013unified, aubry2018new}, which combines the Block Coordinate Descent (BCD) paradigm with the MM framework, a novel solution technique for Problem \eqref{Pk1} is developed. Specifically, along the BCD iterations, a surrogate function for the restriction of the original objective function with respect to each block is optimized. In order to proceed, let $\{\widetilde \bc_{nk}^i\}_{n\in\calN}$ be the optimized codes up to the $i$-th iteration, and let $\widetilde\bc_{pk},p\in\calN$ be the block to be updated at $(i+1)$-th iteration. The restriction of the objective function with respect to the $p$-th block is given by
\beq\label{hat_fpk}
\widehat f_{pk}(\widetilde\bc_{pk}) =\Tr\!\Big(\big(\bB_{pk}^i+\bH_{pk}^\trans\widehat \bS_{pk}^{-1}(\widetilde\bc_{pk})\bH_{pk}\big)^{-1}\Big),
\eeq
where
\begin{align}
\bB_{pk}^i =& (\bF^{-1})^\trans\widetilde\bJ_{k-1}(\widetilde\bc_{1(k-1)},\widetilde\bc_{2(k-1)},\cdots\!,\! \widetilde\bc_{N(k-1)})\bF^{-1} \nonumber\\
&+ \sum\limits_{\substack{ n\in\calN,\\n\neq p}}\bH^\trans_{nk}\widehat\bS_{nk}^{-1}(\widetilde\bc_{nk}^{i})\bH_{nk}.
\end{align}

It is worth noting that $\bB_{pk}^i\in\bbS_{++}^{M}$, since $\widehat\bS_{nk}(\widetilde\bc_{nk}^{i})\in\bbS_{++}^{M}$ and the positive definite initial IM $\bJ_0$.

Following the principles of the block-MM framework, a surrogate function for $\widehat f_{pk}(\widetilde\bc_{pk})$ is now derived. To this end, the following preliminary results are necessary.

\begin{proposition}\label{prop_fRconcave}
 	Let $\bB \in\bbS_{++}^M$, the function $h(\bR)=\Tr\!\big(\big(\bB+\bH^\trans\bR^{-1}\bH\big)^{-1}\big)$ for $\bR\in \bbS_{++}^M$ is a concave function.  
\end{proposition}
 \begin{proof}
	See Appendix \ref{proof_prop_fRconcave} of the supplementary material.
\end{proof}
\begin{lemma}[\!\!\cite{song2015optimization}]\label{lemma:1}
	For matrices $\bA,\bL,\bU\in\bbS^{M\times M}$ such that $ \bA\preceq \bU$ and $ \bA\succeq \bL$, the quadratic form $\bx^\trans \bA \bx$ can upperbounded and lowerbounded for any $\bx_0\in\bbR^N\setminus \{\bzero\}$
	\beq
		\bx^\trans \bA \bx \!\leq \!\bx^\trans \bU \bx \!+  \!2\bx^\trans (\bA \!-\! \bU) \bx_0   \!+ \!\bx_0^\trans (\bU \!-\! \bA) \bx_0\label{lemma:1upper}
	\eeq
	and 
	\beq
		\bx^\trans \bA \bx \!\geq \!\bx^\trans \bL \bx \!+  \!2\bx^\trans (\bA \!-\! \bL) \bx_0   \!+ \!\bx_0^\trans (\bL \!-\! \bA) \bx_0\label{lemma:1lower}
	\eeq
	where the equality is achieved at $\bx=\bx_0$.
\end{lemma}

To proceed further, let
\beq\label{hpk:R}
h_{pk}(\bR_{pk})=\Tr\!\Big(\big(\bB_{pk}^i\!+\!\bH_{pk}^\trans \bR_{pk}^{-1}\bH_{pk}\big)^{-1}\Big), 
\eeq
which is concave according to Proposition \ref{prop_fRconcave}. Based on the first order condition for concave functions at some given $\bR_{pk}^i$, the following inequalities hold: 
\begin{align}\label{fR_maj1}
	&h_{pk}(\bR_{pk}) \nonumber\\
	&\leq\! h_{pk}(\bR_{pk}^{i})+ \Tr\Big(\Big(\frac{\partial h_{pk}(\bR_{pk})}{\partial \bR_{pk}}\Big)^\trans\Big|_{\bR_{pk}\!=\bR_{pk}^{i}}\!\big(\bR_{pk}-\bR_{pk}^{i}\big)\Big) \nonumber\\
	 &=\! h_{pk}(\bR_{pk}^{i})\!+\! \Tr\Big((\bR_{pk}^{i})^{-\!1}\bH_{pk}(\bB_{pk}^i\!+\!\bH_{pk}^\trans(\bR_{pk}^{i})^{-\!1}\bH_{pk}\big)^{-\!2}\nonumber\nonumber\\
	 &\quad \cdot \bH_{pk}^\trans(\bR_{pk}^{i})^{-\!1}(\bR_{pk}\!-\!\bR_{pk}^{i})\Big)\nonumber \\
	 &=\! h_{pk}(\bR_{pk}^{i})\!+\!\Tr(\widehat \bD_{pk}^i(\bR_{pk}-\bR_{pk}^i))\nonumber\\
	 &=\! h_{pk}(\bR_{pk}^{i})\!+\!\sum\limits_{l=1}^{3} [\widehat\bD_{pk}^i]_{ll}[\bR_{pk}]_{ll} \!-\! \Tr(\widehat\bD_{pk}^i\bR_{pk}^i),
\end{align}
where 
\begin{align}\label{D_pki}
&\widehat\bD_{pk}^i\nonumber \\
&\;=\!(\bR_{pk}^i)^{-\!1}\bH_{pk}(\bB_{pk}^i\!+\!\bH_{pk}^\trans(\bR_{pk}^i)^{-\!1}\bH_{pk}\big)^{-\!2}\bH_{pk}^\trans(\bR_{pk}^i)^{-\!1}.
\end{align}
Now, replacing $\bR_{pk}=\widehat \bS_{pk}(\widetilde\bc_{pk})$ in \eqref{hpk:R} yields $\widehat f_{pk}(\widetilde\bc_{pk})=h_{pk}(\bR_{pk})$ on the basis of \eqref{hat_fpk}. By further substituting $\bR_{pk}^i=\widehat \bS_{pk}(\widetilde\bc_{pk}^i)$ into \eqref{fR_maj1}-\eqref{D_pki}, where $\widetilde \bc^i_{pk}$ is the solution at the $i$-th iteration, a majorizing function for $\widehat f_{pk}(\widetilde\bc_{pk})$ at $\widetilde \bc_{pk}=\widetilde \bc^i_{pk}$ is obtained as
\beq\label{f_maj1}
	\widehat f_{pk}(\widetilde\bc_{pk}) \leq \widetilde \bc_{pk}^\trans\bD_{pk}^i \widetilde \bc_{pk} + (\bd_{ pk}^i)^\trans \widetilde \bc_{pk} +  d_{pk}^i,
\eeq
where $\bD_{pk}^i \!\!=\!\!\sum_{l=1}^3[\widehat\bD_{pk}^i]_{ll}\bA_{l, pk}$, $\bd_{pk}^i \!\!=\!\!\sum_{l=1}^3[\widehat\bD_{pk}^i]_{ll}\ba_{l, pk}$, and $d_{pk}^i=\widehat f_{pk}(\widetilde\bc_{pk}^i)\!-\! \Tr(\widehat\bD_{pk}^i\bR_{pk}^i)+\sum_{l=1}^3[\widehat\bD_{pk}^i]_{ll}a_{l, pk}$.

Next, applying \eqref{lemma:1upper} of Lemma \ref{lemma:1} to the quadratic term $\widetilde \bc_{pk}^\trans\bD_{pk}^i \widetilde \bc_{pk}$ in \eqref{f_maj1}, and exploiting the unit energy constraint $\|\widetilde \bc_{pk}\|^2=1$, the objective function can be further tightly upperbounded as
\begin{align}
	\widehat f_{pk}(\widetilde\bc_{pk}) \leq (\bg_{pk}^i)^\trans\widetilde \bc_{pk}+g_{pk}^i,
\end{align}
where 
\begin{align}
	\bg_{pk}^i =& 2(\bD_{pk}^i-\lambda_{\max}(\bD_{pk}^i)\bI_{2M})\widetilde \bc_{pk}^i + \bd_{pk}^i,\label{ggpk}\\
	g_{pk}^i=&2\lambda_{\max}(\bD_{pk}^i) - (\widetilde \bc_{pk}^i)^\trans\bD_{pk}^i\widetilde \bc_{pk}^i +  \widetilde d_{pk, t}.\label{gpk}
\end{align}

As a consequence, $\widetilde\bc_{pk}^{i+1}$ can be obtained by solving the surrogate problem
\beq\label{Pk_maj1}
\!\!\!\left\{\!\!\!\begin{array}{lll}
	\min\limits_{\widetilde\bc_{pk}} & \!\!\!(\bg_{pk}^i)^\trans\widetilde \bc_{pk}+g_{pk}^i\\
	\subto &\!\!\! \|\widetilde\bc_{pk}\|^2 = 1\\
	&\!\!\! \widetilde \bc_0^\trans\widetilde\bc_{pk}\geq \zeta_1\\
	&\!\!\! \widetilde \bc_{pk}^\trans\bA_{l, pk} \widetilde \bc_{pk} \!+\! \ba_{l, pk}^\trans 	\widetilde \bc_{pk} \!+\! a_{l, pk}\!\geq\! \varepsilon,l=1,2,3\!\!\!\!
\end{array}.
\right.
\eeq

Problem \eqref{Pk_maj1} remains non-convex and NP-hard in general due to the energy requirement and the last non-convex quadratic constraints. However, the latter can be linearized into a convex form by invoking Lemma \ref{lemma:1} along with the unit energy condition. Specifically, using \eqref{lemma:1lower} from Lemma \ref{lemma:1} and $\|\bc_{pk}\|^2 = 1$, a tight lowerbound for $\widetilde \bc_{pk}^\trans\bA_{l, pk} \widetilde \bc_{pk}$ is 
\begin{align}\label{cAc_approx}
	\widetilde \bc_{pk}^\trans\bA_{l, pk} \widetilde \bc_{pk} \geq&    2\widetilde\bc_{pk}^\trans (\bA_{l,pk}-\lambda_{\min}(\bA_{l,pk})\bI_{2M}) \widetilde\bc_{pk}^i\nonumber\\
	&2\lambda_{\min}(\bA_{l,pk})-(\widetilde\bc_{pk}^i)^\trans\bA_{l,pk}\widetilde\bc_{pk}^i,
\end{align}
where the equality holds when $\widetilde \bc_{pk}=\widetilde \bc_{pk}^i$. Thus, the last quadratic constraints in Problem \eqref{Pk_maj1} can be approximated by the following convex constraints
\beq
(\bh_{l, pk}^i)^\trans\widetilde\bc_{pk}+  h_{l,pk}^i\geq\varepsilon,l=1,2,3,
\eeq
where 
\begin{align}
	\bh_{l, pk}^i=&2(\bA_{l, pk}-\lambda_{\min}(\bA_{l, pk})\bI_{2M})\widetilde\bc_{pk}^i+\ba_{l, pk},\label{hhpk}\\
	 h_{l,pk}^i=&2\lambda_{\min}(\bA_{l, pk})-(\widetilde\bc_{pk}^i)^\trans\bA_{l, pk}\widetilde\bc_{pk}^i+a_{l, pk}.\label{hpk}
\end{align} 

Therefore, Problem \eqref{Pk_maj1} can be cast as
\beq\label{Pk_sca}
 \!\!\!\left\{\!\!\!\begin{array}{lll}
 	\min\limits_{\widetilde\bc_{pk}} & \!\!\!(\bg_{pk}^i)^\trans\widetilde \bc_{pk}+g_{pk}^i\\
 	\subto &\!\!\! \|\widetilde\bc_{pk}\|^2 = 1\\
 	&\!\!\! \widetilde \bc_0^\trans\widetilde\bc_{pk}\geq \zeta_1\\
 	&\!\!\! (\bh_{l, pk}^i)^\trans\widetilde\bc_{pk}+  h_{l,pk}^i\geq\varepsilon,l=1,2,3
 \end{array}.
 \right.
\eeq

Problem \eqref{Pk_sca} is feasible\footnote{According to the tight lowerbound argument of Lemma \ref{lemma:1}, $\widetilde \bc_{pk}^i$ is a feasible point to Problem \eqref{Pk_sca}.} and its optimal solution can be obtained applying the following lemma, which reveals the hidden convexity of Problem \eqref{Pk_sca}.

\begin{lemma}\label{lemma_Pk_sca_realx}
	The optimal solution to Problem \eqref{Pk_sca} can be constructed from the optimal solution to Problem \eqref{Pk_sca_relax}.
	\beq\label{Pk_sca_relax}
	\!\!\!\left\{\!\!\!\begin{array}{lll}
		\min\limits_{\widetilde\bc_{pk}} & \!\!\!(\bg_{pk}^i)^\trans\widetilde \bc_{pk}\\
		\subto &\!\!\! \|\widetilde\bc_{pk}\|^2 \leq 1\\
		&\!\!\! \widetilde \bc_0^\trans\widetilde\bc_{pk}\geq \zeta_1\\
		&\!\!\! (\bh_{l, pk}^i)^\trans\widetilde\bc_{pk}+ h_{l,pk}^i\geq\varepsilon,l=1,2,3
	\end{array}.
	\right.
	\eeq
\end{lemma}
\begin{proof}
	See Appendix \ref{proof_Pk_sca_relax} of the supplemental material.
\end{proof}

Therefore, the convex Problem \eqref{Pk_sca_relax} can be efficiently solved using a convex optimization solver such as CVX\cite{grant2020cvx}. An optimal solution that satisfies the energy equality constraint, can then be readily constructed by following the procedure described in \eqref{opt_c_construct} and \eqref{obj_diff} of Appendix \ref{proof_Pk_sca_relax} in the supplemental material.

The final form of the designed block-MM based procedure to tackle $\bar \calP_k$, namely $\widehat \calP_k$, is summarized in Algorithm \ref{algo2}. 
The convergence properties of the proposed block-MM algorithm for solving $\bar \calP_k$ are established in the following proposition.
\begin{proposition}\label{proposition:3}
	Let $\widetilde \bc^i=[(\widetilde \bc_{1k}^{i})^\trans,(\widetilde \bc_{2k}^{i})^\trans,\cdots,(\widetilde \bc_{Nk}^{i})^\trans]^\trans$ be a sequence of points obtained through the block-MM algorithm, and $\widetilde f_k^i=f_k(\widetilde\bc_{1k}^i,\widetilde\bc_{2k}^i,\cdots\!,\! \widetilde\bc_{Nk}^i)$ be the objective value at the $i$-iteration. Then, $\widetilde f_k^i$ is a monotonically increasing sequence, converging to a finite value.
\end{proposition}
\begin{proof}
	See Appendix \ref{proof_propostion3} of the supplemental material.
\end{proof}

\begin{algorithm}
	\setstretch{1.1}
	\caption{The block-MM framework to solve $\widehat\calP_k$.}\label{algo2} 
	\algorithmicrequire{ \multiline{ $k$, $\widetilde \bc_{0}$, $\zeta$, $\bF$, $\widetilde\bJ_{k-1}$, $\bH_{nk}$, $\bA_{l,nk}$, $\ba_{l,nk}$, $a_{l,nk},n\!\in\!\calN$, \\$l\in\{1,2,3\}$, $\xi$.}} \\
	\algorithmicensure{ Optimized solutions $\{\widetilde\bc_{nk}^\star\}_{n\in\calN}$.} 
	\begin{algorithmic}[1]
		\State Let $i=1$ and $\widetilde \bc_{nk}^0=\widetilde\bc_0, n\in\calN$;
		\While{$|\widehat f_k(\{\widetilde\bc_{nk}^i\}_{n\in\calN})-\widehat f_k(\{\widetilde\bc_{nk}^{(i-1)}\}_{n\in\calN})|\geq\xi$}
		\For {$p=1,2,\cdots,N$}
		\State \multiline{Compute $\bg_{pk}^i, g_{pk}^i$ as defined in \eqref{ggpk}-\eqref{gpk};}
		\State \multiline{Compute $\bh_{l,pk}^i, h_{l,pk}^i$ as given in \eqref{hhpk}-\eqref{hpk};}
		\State Obtain $\widetilde\bc_{pk}^i$ via solving Problem \eqref{Pk_sca_relax};
		\EndFor
		\State $i=i+1$;
		\EndWhile
		\State Output $\widetilde\bc_{nk}^\star=\widetilde\bc_{nk}^i, n\in\calN$.
	\end{algorithmic}
\end{algorithm}
\section{Numerical  Results} \label{SecIV}
Numerical experiments are conducted to evaluate the performance of the proposed method for radar network code design. Consider a radar network with $N=4$ nodes, operating at $X$ band (the carrier frequency is 10 GHz), using frequency orthogonal waveforms. The antenna of each node consists of a uniform linear array of $N_\rmr =$ 8 elements spaced at half-wavelength intervals. Their transmitted trains utilize $M =$ 8 pulses with a PRI of $T_\rmr = $ 250 ${\upmu}$s. The chirp signal has a bandwidth of $B=$ 5 MHz and a pulsewidth of $T_\rmp=$ 10 $ \upmu$s. The sampling time is $\Delta t=1/2B$ leading to $N_\rmp=$ 100. The observation sampling interval between consecutive tracking frames is set to $T=$ 1 s. 

Consider a single point-like target moving in the radar system surveillance area, as depicted in Fig. \ref{f1}. The four radar nodes are located at (20, 10)~km, (25, 16)~km, (35, 16)~km, and (40, 10)~km, respectively. It is assumed that the antenna beams of the radar nodes consistently point towards the target. The initial location and velocity of the target are (30, 55) km and (80, 240) m/s, respectively. The target power $|\alpha_{nk}|^2$ is assumed to remain constant across all frames, and are set as $|\alpha_{1k}|^2=$ 0.035, $|\alpha_{2k}|^2=$ 0.099, $|\alpha_{3k}|^2=$ 0.176, $|\alpha_{4k}|^2=$ 0.051.  The signal-independent interference covariance matrix\cite{de2008code, de2009code} in the slow-time domain ${\bSigma_{ \rmt, nk}}$, and spatial domain ${\bSigma_{ \rms, nk}}$ are modeled as ${\bSigma_{ \rmt, nk}}(i, j) = \rho_\rmt^{|i-j|}$, ${\bSigma_{ \rms, nk}}(i, j) = \rho_\rms^{|i-j|},n\in\calN$, where the one-lag correlation coefficients are $\rho_\rmt =$ 0.8 and $\rho_\rms=$ 0.8, respectively. The similarity sequence is considered as a P3 code \cite[pp. 118-122]{levanon2004radar}.

The initial IM $\bJ_{0}$ is initialized as a diagonal matrix whose diagonal entries are all $10^{-10}$. The regularization parameter for the approximation $\widehat \bS_{nk}(\widetilde\bc_{nk})$ is set to $\varepsilon=10^{-8}$. The stopping parameter of the developed block-MM procedure is fixed to $\xi=10^{-3}$.
\subsection{Analysis of Convergence}
This subsection analyzes the convergence of the approximated PCRLB within a frame, as well as the PCRLB behavior across frames achieved by the block-MM algorithm. 

\begin{figure}[htbp]
	\subfigure[]{\label{f3}
		\includegraphics[width=1.58in]{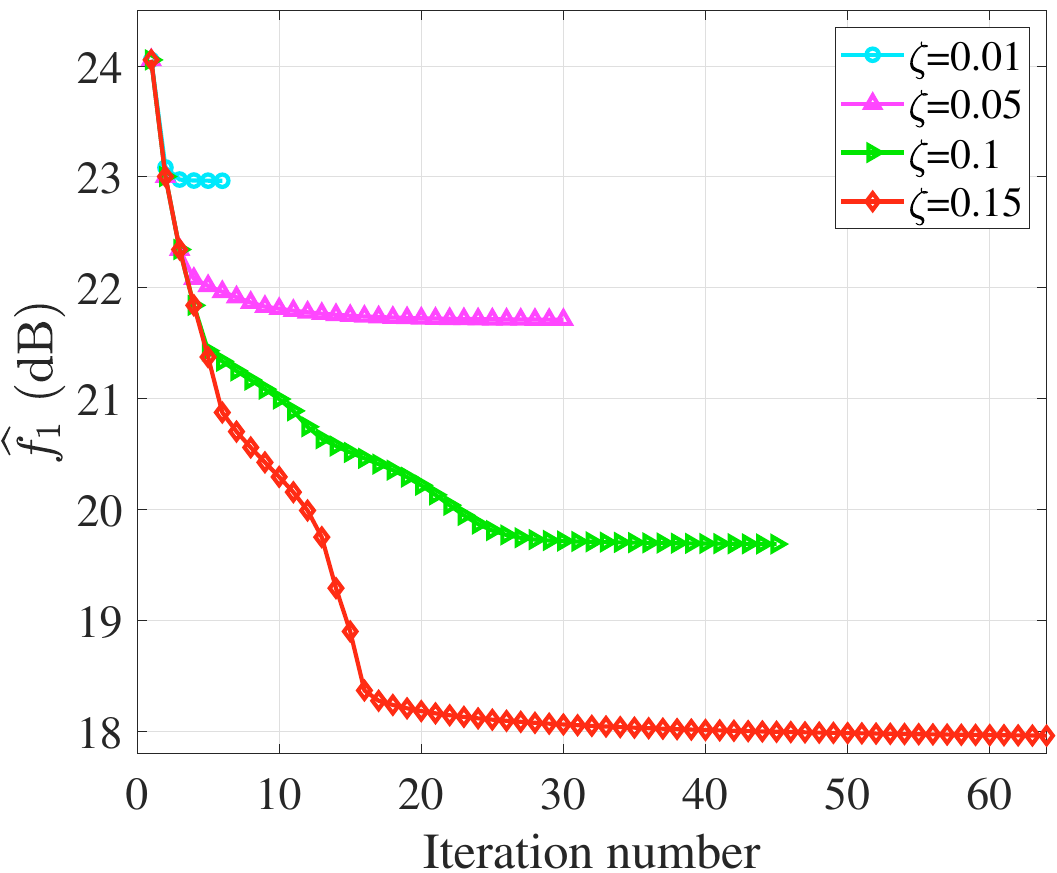}
	}\hspace{-3mm}
	\subfigure[]{\label{f4a}
		\includegraphics[width=1.58in]{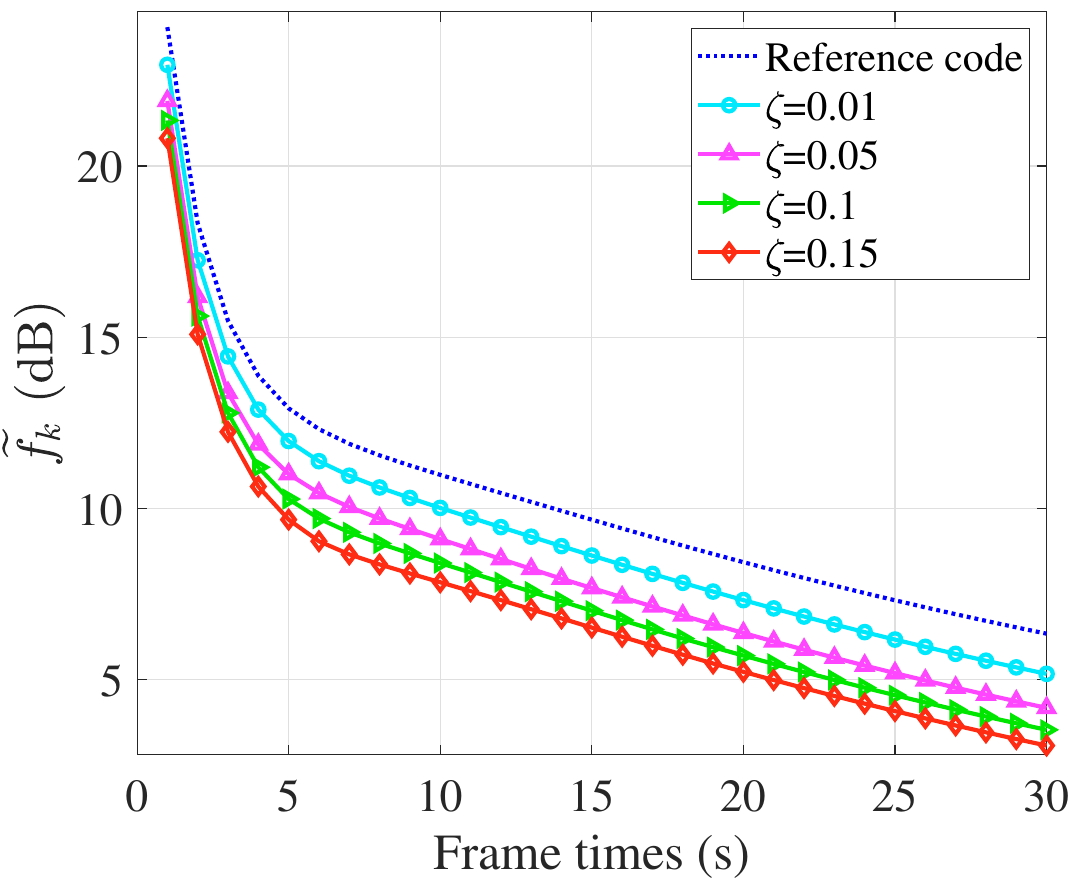}
	}
	\centering
	\caption{\small (a) The approximated PCRLB performances at the first frame versus iteration number. (b) The PCRLB performances versus frame time achieved by the designed codes and the reference sequence.}\label{f4}
\end{figure}

At first, the approximated PCRLB in $\widehat P_k$ is assessed, using the first frame (i.e., $k=1$) as an illustrative example. Fig. \ref{f3} plots the approximated objective $\widehat f_k$ in \eqref{approx_fk} versus iterations for $\zeta=$ 0.01, 0.05, 0.1, 0.15. Inspection of the curves highlight that the approximated objective value monotonically decreases along the iterations. Besides, the larger the similarity parameter $\zeta$ the lower the achieved objective values due to the enlarged feasible set of $\widehat P_1$.

In Fig. \ref{f4a}, the PCRLB trace $\widetilde f_k$ is shown versus the frame time for the designed sequences with distinct values of $\zeta$, respectively. For comparison, Fig. \ref{f4a} also reports the PCRLB achieved using the reference code throughout all frames. Inspection of the curves highlights that $\widetilde f_k$ decreases monotonically over the frames, and the designed sequences consistently outperform the reference code. In particular, at the 30-th frame, the codes obtained under $\zeta=$ 0.15 realize a remarkable PCRLB gain of 3.5 dB compared to the reference. 


\begin{figure}
	\centering
	\includegraphics[width=2.8in]{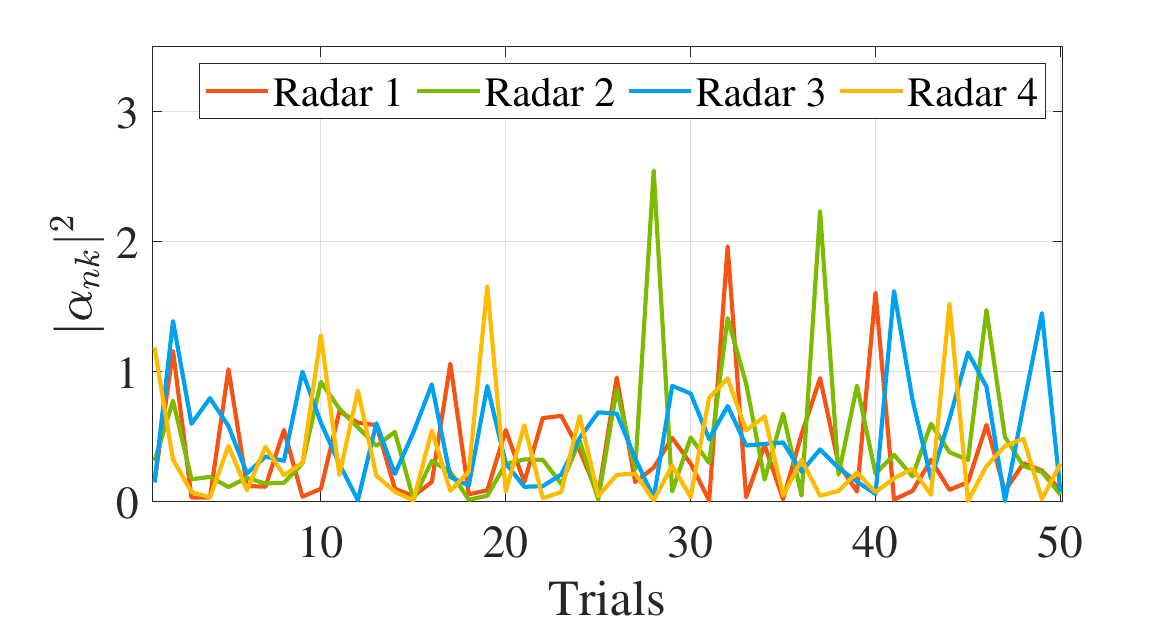}
	\caption{\small Target power at the four radar nodes over 50 trials.}\label{f5}
\end{figure}
\subsection{Analysis of Detection Probability and Target State Estimation Performance}
In this section, 50 Monte Carlo trials with different values of $|\alpha_{nk}|^2,n\in\calN$, are conducted to evaluate the target detection probability and target state estimation accuracy. For each trial and each radar node, a value of $|\alpha_{nk}|^2$ is independently drawn from an exponential distribution with parameter 0.5 and kept constant across all frames of that trial. The values of $|\alpha_{nk}|^2$ for the four nodes over all trials are shown in Fig. \ref{f5}. In each trial, the PCRLB is optimized using the proposed method with the corresponding $|\alpha_{nk}|^2$. The  final performance is obtained by averaging the results across all trials.

\subsubsection{Detection Probability}As a benchmark for target detection, let us consider only the optimization of $P_{\rmd,nk}$ under the energy and similarity constraints. According to the expression of $P_{\rmd, nk}$ in \eqref{Pd}, its maximization can be accomplished by optimizing the SINR over the radar code\cite{de2008code}. In this context, the waveform design for radar nodes is separable, allowing the code synthesis problem for the $n$-th radar to be formulated as 
\beq \label{onlySINRopt}
\left\{\begin{array}{lll}
	\min\limits_{\bc_{nk}} & \bc_{nk}^\ctrans\bM_{0,nk}\bc_{nk}\\
	\subto & \| \bc_{nk}\|^2 = 1,n\in\calN\\
	& \| \bc_{nk}-\bc_0\|^2 \leq \zeta,n\in\calN
\end{array}.
\right.
\eeq 
The optimal solution to Problem \eqref{onlySINRopt} can be found using the procedure outlined in \cite{li2006signal}. The corresponding detection probability for each radar node is used as a benchmark upper bound, denoted by $P_{\rmd, nk}$.

\begin{figure*}[]
	\subfigure[]{\label{f6a}
		\includegraphics[width=1.5in]{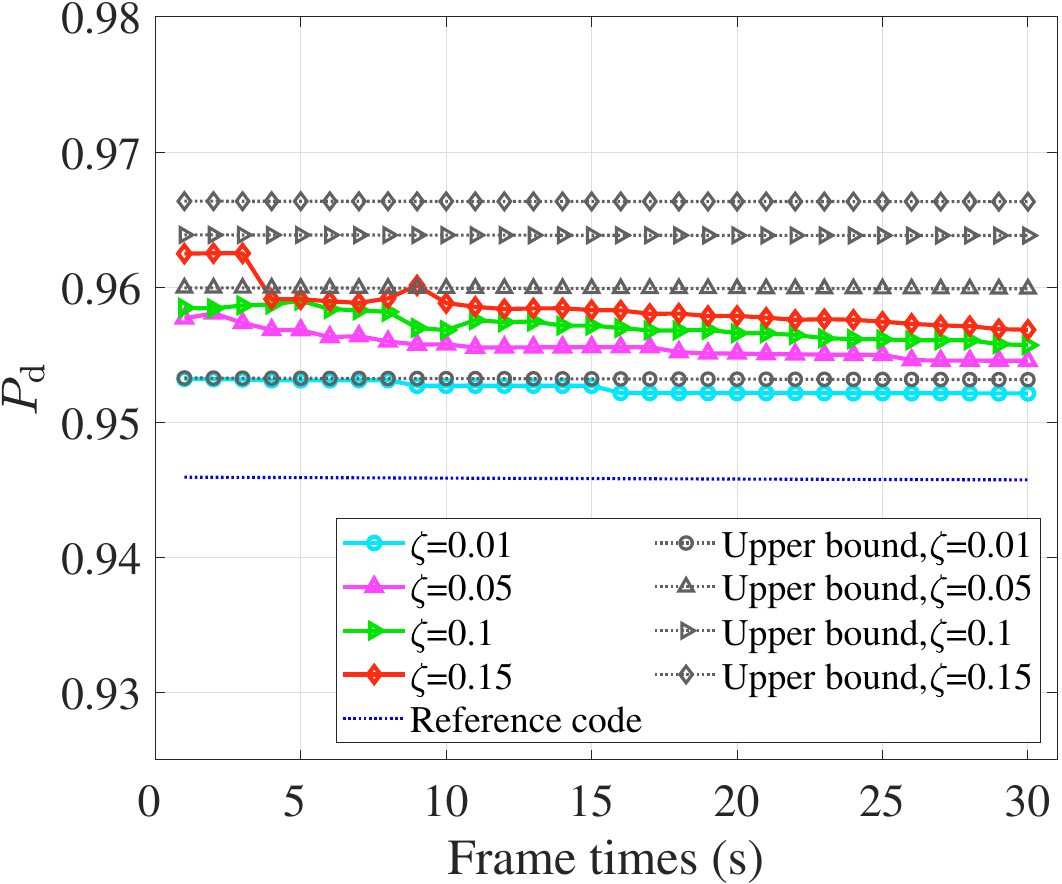}
	}\hspace*{-3mm}
	\subfigure[]{\label{f6b}
		\includegraphics[width=1.5in]{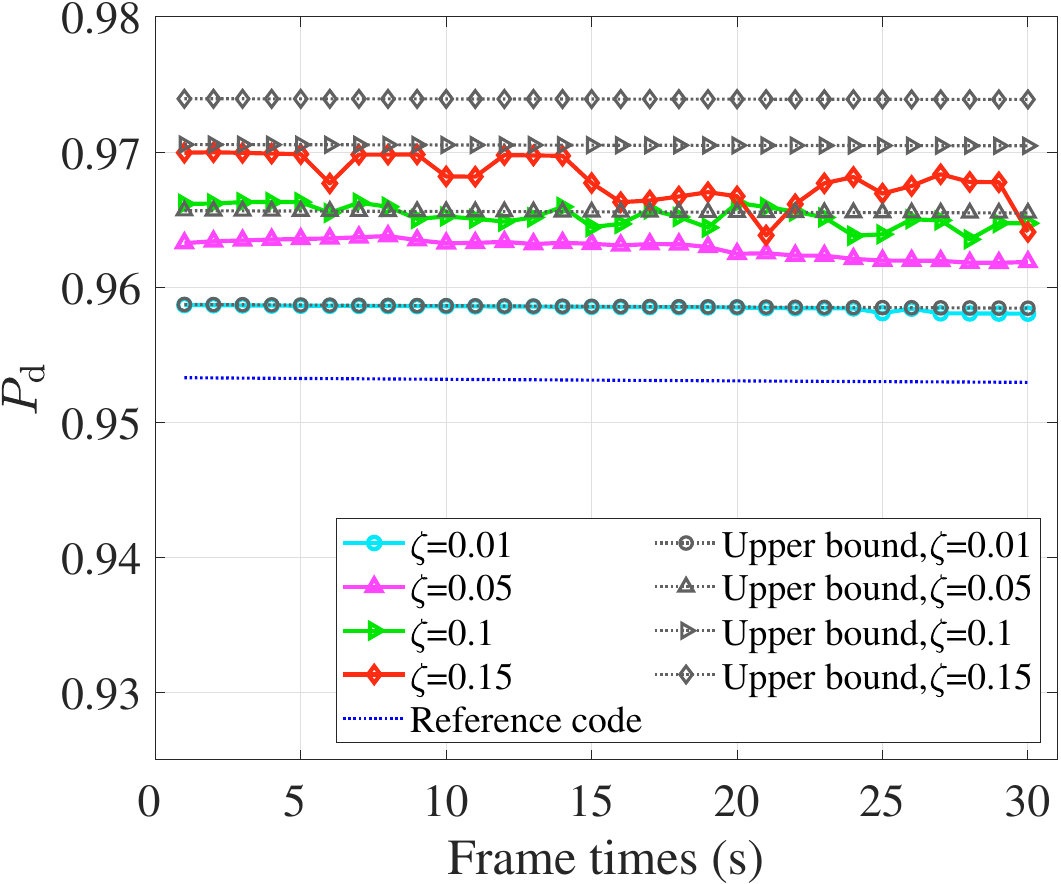}
	}\hspace*{-3mm}
	\subfigure[]{\label{f6c}
		\includegraphics[width=1.5in]{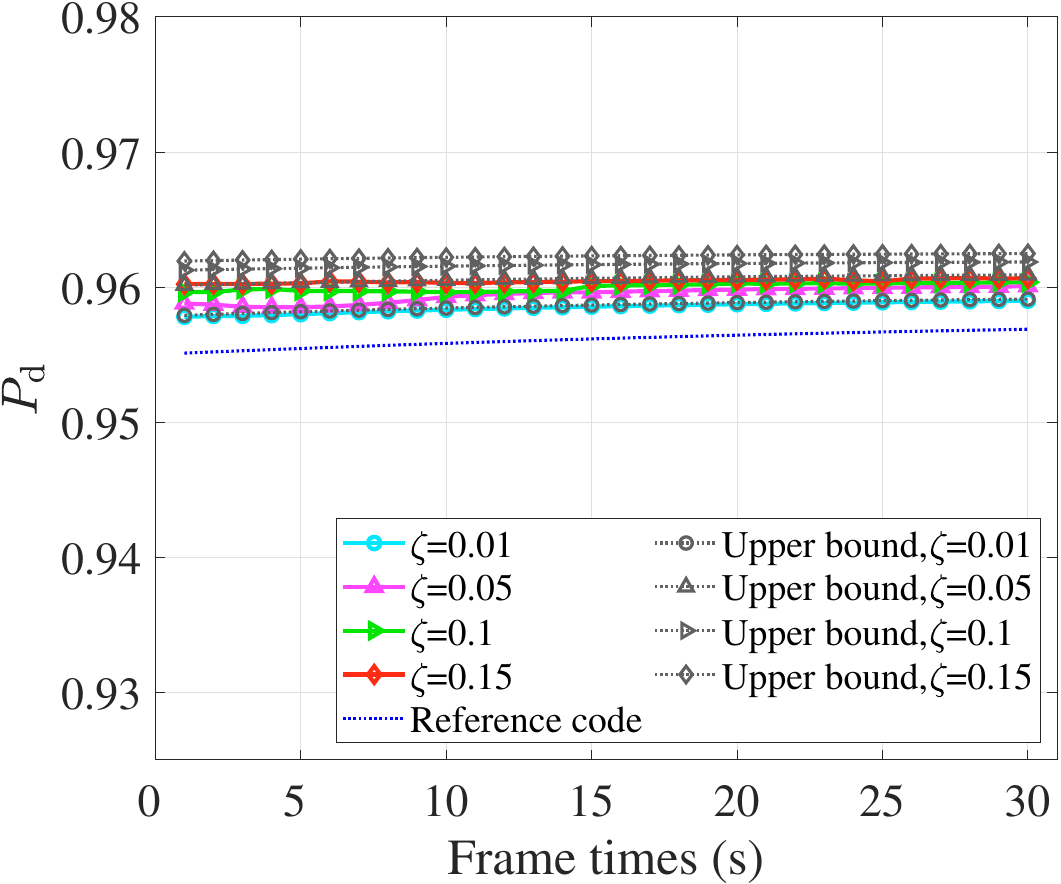}
	}\hspace*{-3mm}
	\subfigure[]{\label{f6d}
		\includegraphics[width=1.5in]{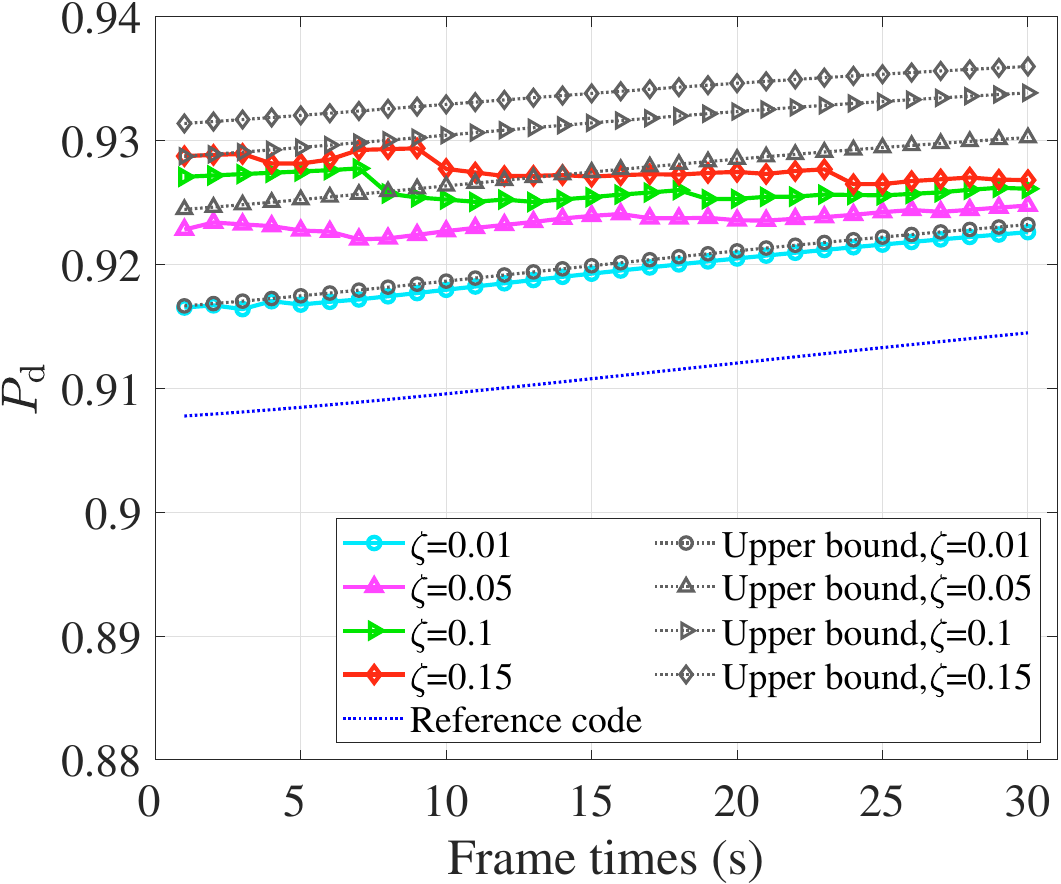}
	}
	\centering
	\caption{\small Average target detection probability versus frame time over 50 trials for the code designed with $\zeta=$ 0.01, 0.05, 0.1, 0.15, their upper bounds and the reference sequence: (a) radar 1, (b) radar 2, (c) radar 3, and (d) radar 4.}\label{f6}
\end{figure*}

Figures \ref{f6}\subref{f6a}, \subref{f6b}, \subref{f6c}, and \subref{f6d} reports the detection probabilities averaged over the 50 Monte Carlo trials for each node with multiple values of the similarity parameter, along with their corresponding benchmark upper bounds and the performance achieved by the reference code. The results indicate that the detection performance achieved by the proposed codes exceeds that of the reference sequence across all tracking frames. In addition, the obtained $P_{\rmd,nk}$ values remain below the benchmark upper bounds (for the same similarity restrictions). This is because our design paradigm optimizes the target state estimation accuracy, taking both detection performance and measurement estimation accuracy into account, rather than focusing solely on maximizing the detection probability. In addition, most detection probabilities increase as $\zeta$ increases due to the larger feasible set. However, for radar 2, the $P_{\rmd,nk}$ at frames 21 and 30, achieved by the code designed with $\zeta=$ 0.15 is lower than that with $\zeta=$ 0.1. This discrepancy arises because our figure of merit is the PCRLB, rather than solely the detection.
\begin{figure}[htbp]
	\centering
	\includegraphics[width=2.7in]{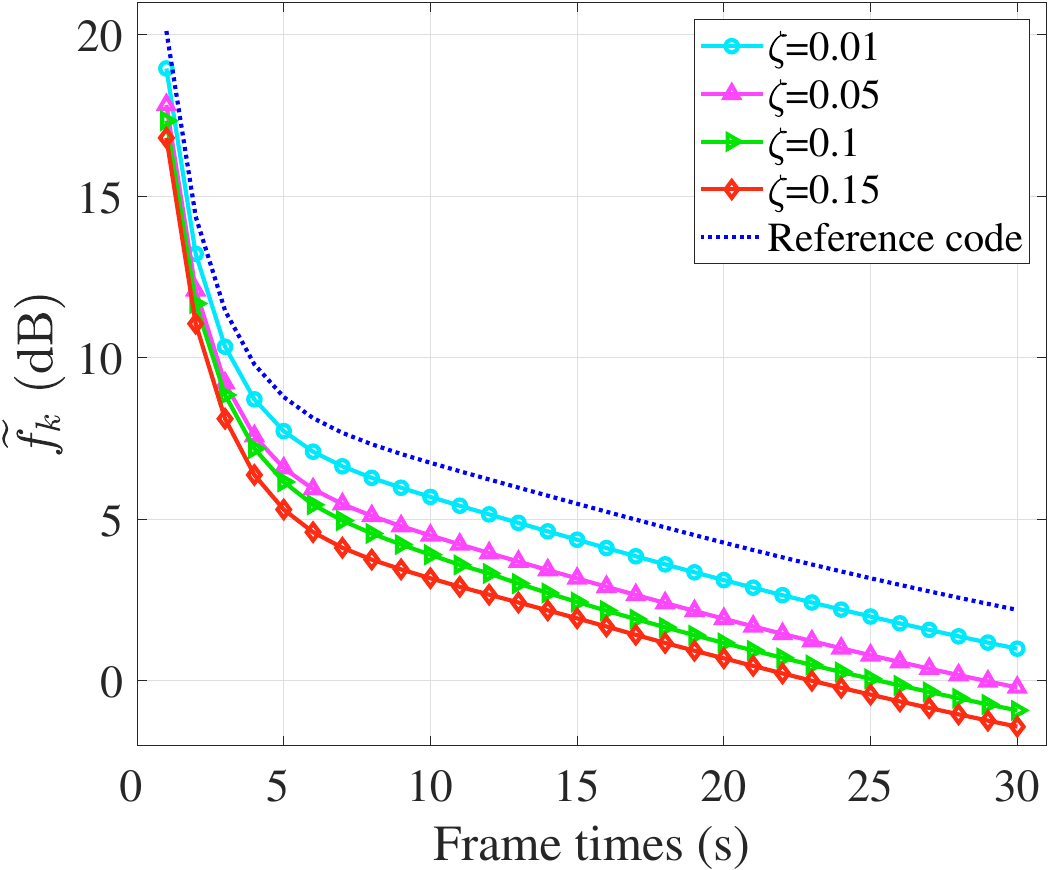}
	\caption{\small Average PCRLB trace versus frame time over 50 trials for the code designed with $\zeta=$ 0.01, 0.05, 0.1, 0.15, and the reference sequence.}\label{f7}
\end{figure}

\begin{figure*}[htbp]
	\subfigure[]{\label{f8a}
		\includegraphics[width=1.53in]{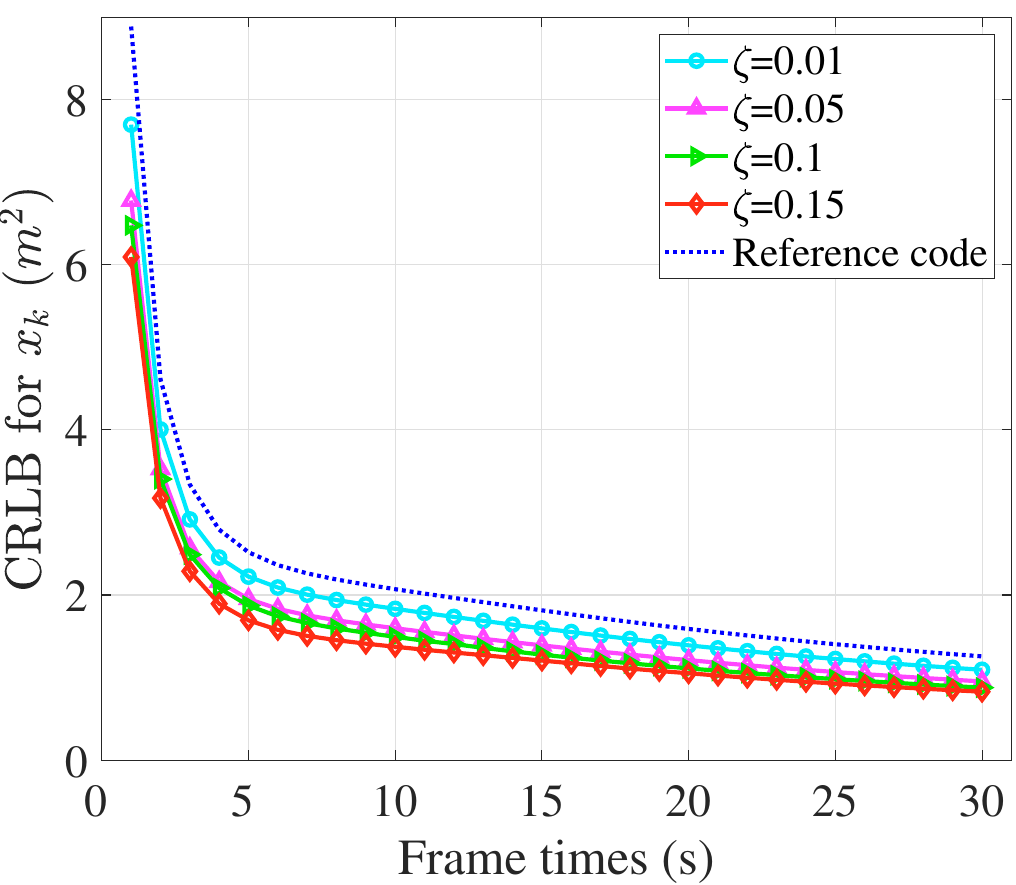}
	}\hspace*{-2mm}
	\subfigure[]{\label{f8b}
		\includegraphics[width=1.6in]{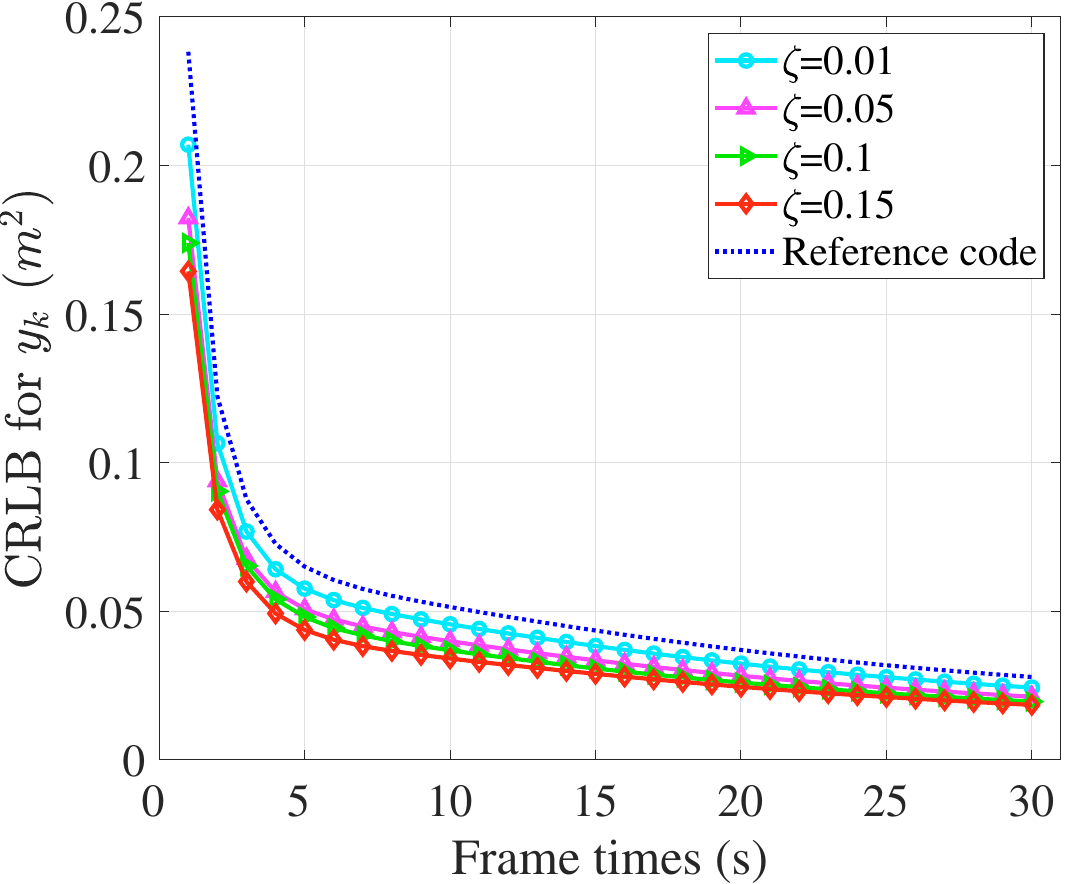}
	}\hspace*{-2mm}
	\subfigure[]{\label{f8c}
		\includegraphics[width=1.6in]{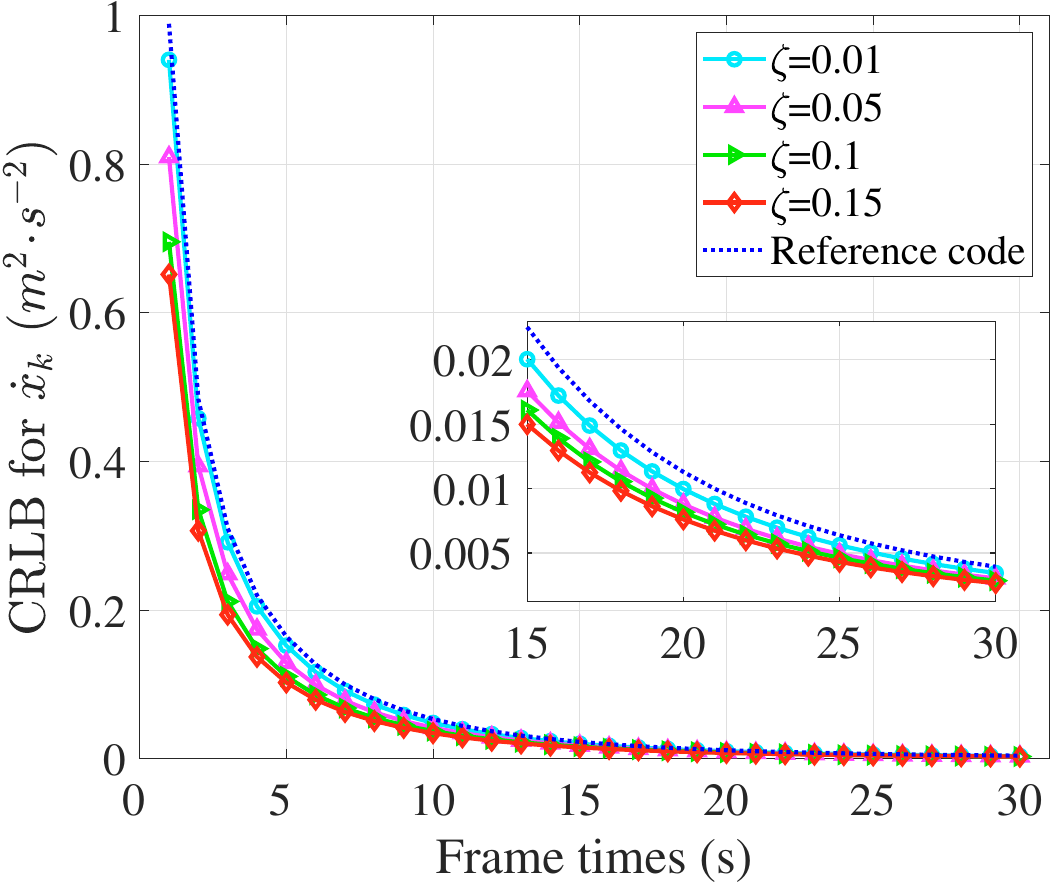}
	}\hspace*{-2mm}
	\subfigure[]{\label{f8d}
		\includegraphics[width=1.6in]{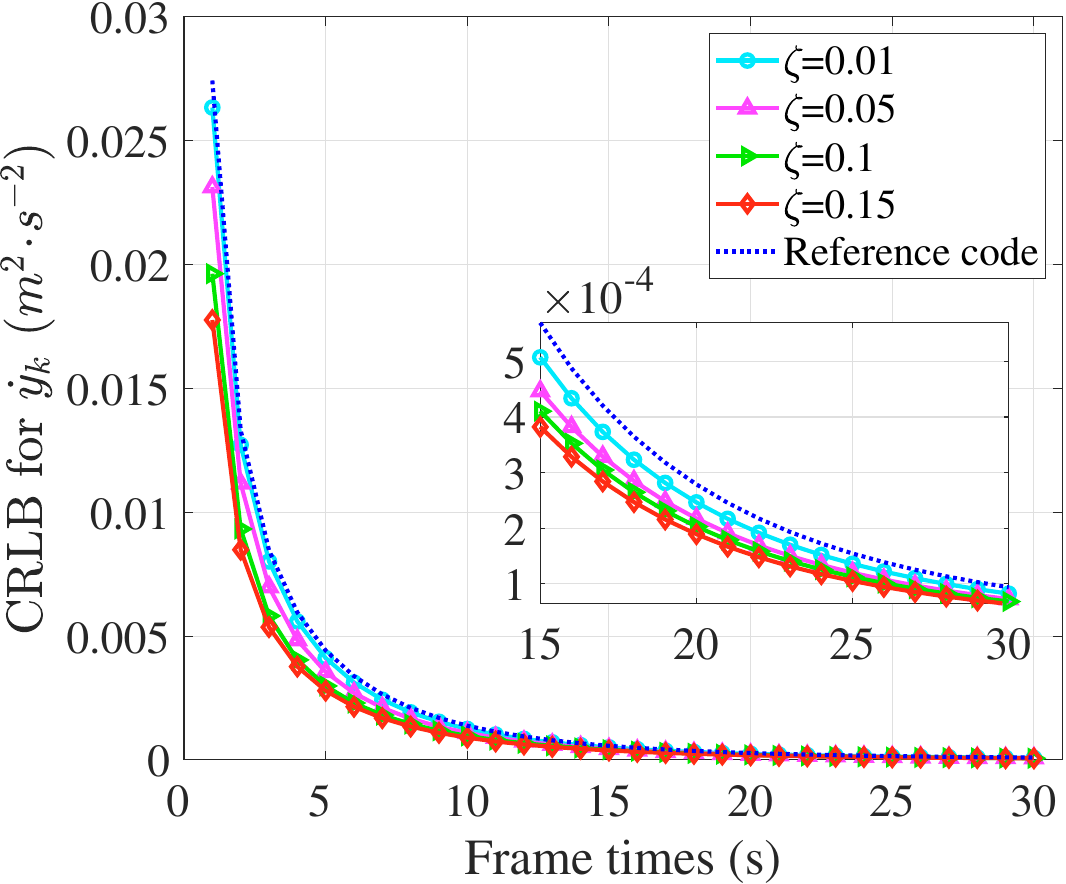}
	}
	\centering
	\caption{\small Average CRLB for the target state estimation versus frame time over 50 trials for the code designed with $\zeta=$ 0.01, 0.05, 0.1, 0.15, and the reference sequence: (a) $x_k$  (${\rm m}^{2}$), (b) $y_k$ (${\rm m}^{2}$), (c) $\dot x_k$ (${\rm m}^{2}\!\cdot\!{\rm s}^{-2}$), and (d) $\dot y_k$ (${\rm m}^{2}\!\cdot\!{\rm s}^{-2}$).}\label{f8}
\end{figure*}
\subsubsection{Target State Estimation}In Fig. \ref{f7}, the frame-by-frame PCRLB trace, averaged over 50 Monte Carlo trials, is shown for codes designed with different values of $\zeta$ as well as for the reference sequence. The curves show that the designed code outperforms the reference over all frames. Furthermore, as expected, increasing either the number of frames or the parameter $\zeta$ results in a decrease in the PCRLB trace. The CRLB for the estimation of $x_k, y_k, \dot x_k, \dot y_k$ are depicted in Figs. \ref{f8a}, \subref{f8b}, \subref{f8c}, and \subref{f8d}, respectively, using the reference and the optimized radar codes assuming $\zeta=$ 0.01, 0.05, 0.1, 0.15. Inspection of the curves emphasizes that the estimation accuracy of the synthesized code is significantly better than that of the reference code over all the frames. Additionally, the estimation accuracies associated with the $x$-axis are worse than those along the $y$-axis. This is because the distance between the target and the radar node along the $x$-axis is smaller than that on the $y$-axis.

\subsection{Robustness Analysis}
\begin{figure}
	\subfigure[]{\label{f9a}
		\includegraphics[width=2.7in]{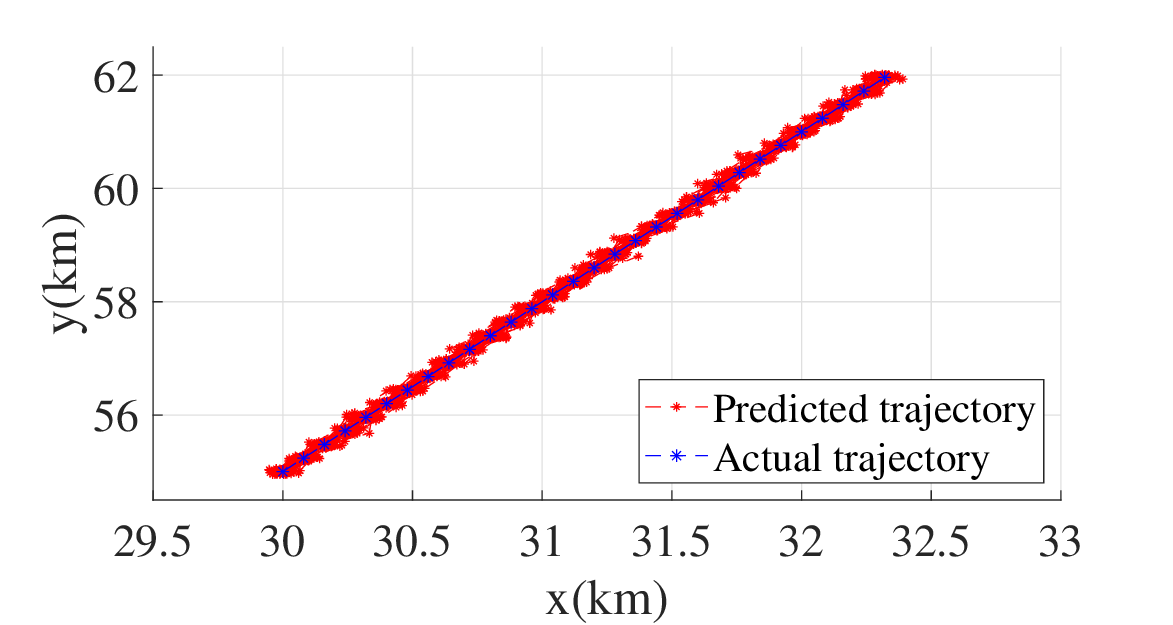}
	}
	\vspace*{-3mm}
	\subfigure[]{\label{f9b}
		\includegraphics[width=3in]{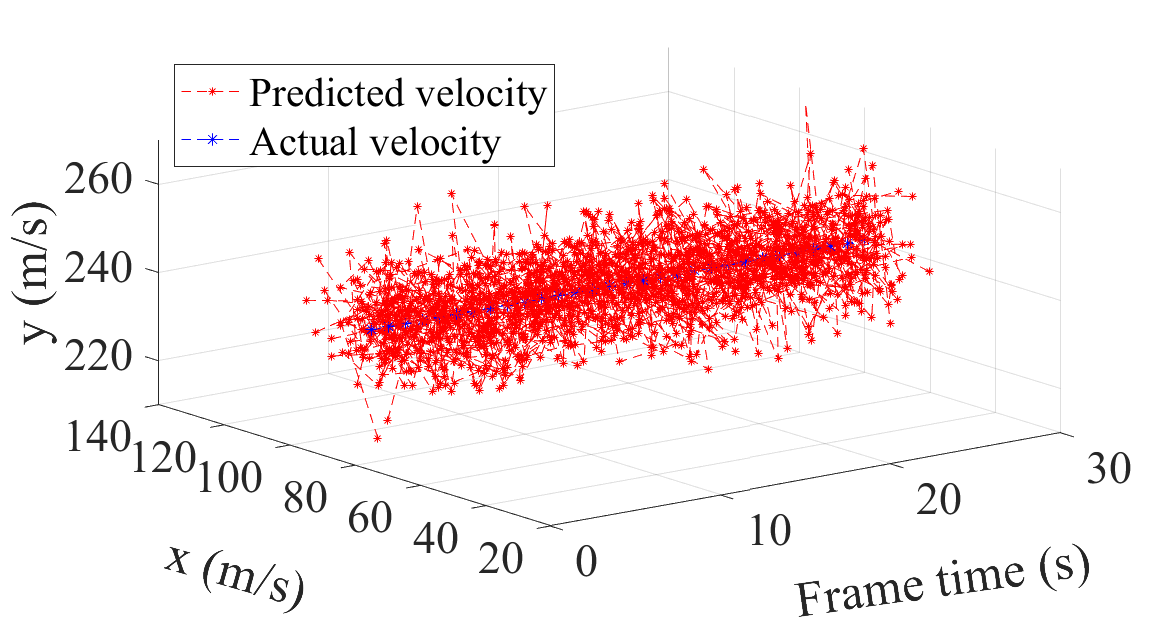}
	}
	\centering
	\caption{\small Predicted target state and actual state over 50 trials: (a) trajectory and (b) velocity.}\label{f9}
\end{figure}

\begin{figure}
	\subfigure[]{
		\includegraphics[width=2.7in]{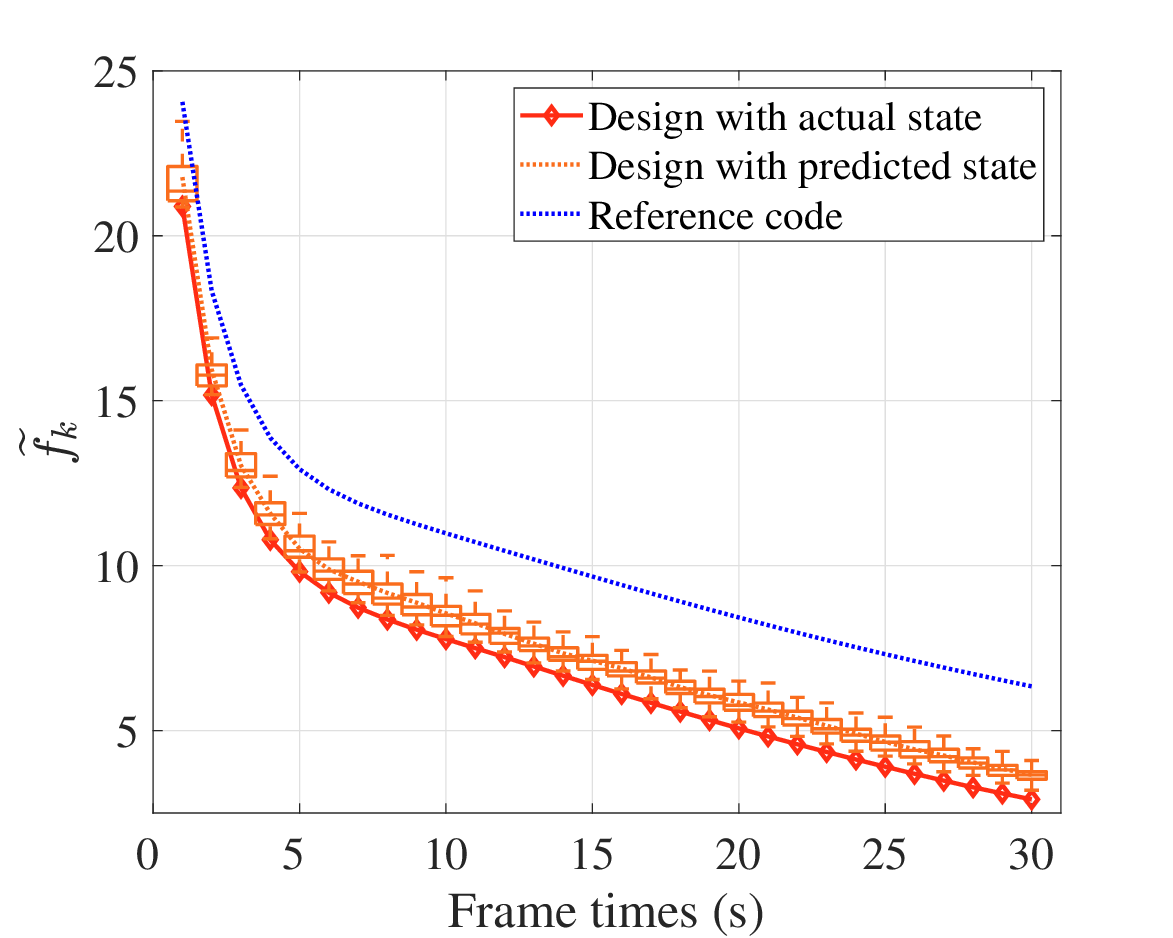}
	}
	\centering
	\caption{\small Average PCRLB with inaccurate predicted target state.}\label{f10}
\end{figure}

The developed design procedure relies on the predicted target state at the next frame. This subsection examines the behavior of the procedure in the presence of a mismatch between the predicted (nominal) and actual target states. The predicted target state $\hat \bx_k$ is assumed to follow a multivariate distribution with mean vector $\bx_k$ and covariance matrix $\bSigma_k\in\bbS^{4}_{++}$, which is given by
\beq
	\bSigma_k = \left[\!\!\bea 
	\sigma_\rmr^2\bI_2 & \bzero_{2\times 2} \\
	\bzero_{2\times 2} & \sigma_\rmv^2\bI_2
	\eea\!\!\right],
\eeq
where $\sigma_\rmr^2$ and $\sigma_\rmv^2$ represent covariances for target position and velocity, respectively. The radar codes designed based on the predicted target state are used to evaluate the mismatch loss when applied to the actual target state. The average PCRLB is adopted as performance metric, and is computed over 50 Monte Carlo trials on the predicted target state. For these tests, the values of $\sigma_\rmr^2$ and $\sigma_\rmv^2$ are set to the squares of the range and velocity resolutions, i.e., 900 $m^2$ and 56.25 $m^2\!\cdot \!s^{-2}$, respectively. The predicted and actual target trajectories and velocities across the 50 trials are shown in Fig. \ref{f9a} and Fig. \ref{f9b}. 

Fig. \ref{f10} depicts the average PCRLB achieved by the devised code with $\zeta=$ 0.15 obtained using the predicted state as well as the actual state. Moreover, the average PCRLB for the reference code, associated with the actual target state, is also plotted as a term of comparison. The box for each frame represents the range of the  trace PCRLB trace values obtained across the conducted trials, while the curves illustrate the corresponding average values. Under the given uncertainty in the target state, the average PCRLB performance is close to that of the error-free condition exhibiting only minor fluctuations. The curves also emphasize that the devised code significantly outperforms the P3 sequence in terms of target state estimation accuracy.

\section{Conclusion} \label{SecV}
This paper presented a novel approach for synthesizing slow-time coded waveforms for single target tracking in radar networks operating under colored Gaussian interference. The optimization framework was based on minimizing the trace of the PCRLB for target state estimation while adhering to power constraints and waveform similarity requirements. By approximating the objective function leveraging the second order Taylor expansion and employing a bespoke block-MM algorithm, an efficient approximation solution technique with guaranteed convergence was developed. The optimized waveforms monotonically reduce the PCRLB trace over frames, and improve target state estimation accuracy compared to reference sequences. Furthermore, the design demonstrates robustness to moderate mismatches in the assumed target state, highlighting its practical effectiveness. Potential future research avenues might include the joint optimization of radar positions \cite{aubry2023robust} and codes in the radar network, as well as the extension of our work to face with signal-dependent cluttered environment\cite{fan2024airborne}.

	\section*{Acknowledgment}
	This research activity has been conducted during the visit of Tao Fan at the Universit\`a degli Studi di Napoli “Federico II”, DIETI under the local supervision of Prof. A. Aubry.  and  Prof. A. De Maio.

\bibliography{./lib/TaoWaveformRefs, ./lib/TaoWaveformForTrackingRefs}

\end{document}


\title{Supplementary materials of ``Radar Network Waveform Design for Target Tracking"
}


\maketitle
\appendix
\subsection{Proof of \eqref{CRLB}}\label{proofCRLB}

As to the signal model in \eqref{echo}, let us consider the estimation of $\bgamma_{nk}=[\widetilde \balpha_{nk}^\trans, \tau_{nk}, f_{\rmd, nk},  \theta_{nk}]^\trans\in\bbR^5$, where $\widetilde \balpha_{nk}\!=\!\big[\Re\{\alpha_{nk}\}, \Im\{\alpha_{nk}\}\big]^\trans \in\bbR^2$ is a nuisance parameter. Denoting by $\bmu({\bgamma_{nk}})\!=\! \alpha_{nk}(\ba_{\rmt}(f_{\rmd, nk})\odot \bc_{nk}) \otimes \bs(\tau_{nk}) \otimes \ba_{\rms} (\theta_{nk})$, the FIM for estimating $\bgamma_{nk}$ is \cite{dogandzic2001cramer}
\beq
\mathbfcal{I} = 2\Re\big\{\bPsi^\ctrans(\bgamma_{nk})\bSigma_{{\rmv}, nk}^{-1}\bPsi(\bgamma_{nk})\big\},\label{FIM}
\eeq
where 
\begin{align}
	\bPsi(\bgamma_{nk}) &= \frac{\partial \bmu({\bgamma_{nk}})}{\partial \bgamma_{nk}^\trans}\nonumber\\
	&=\bigg[\frac{\partial \bmu({\bgamma_{nk}})}{\partial \widetilde \balpha_{nk}^\trans},\frac{\partial \bmu({\bgamma_{nk}})}{\partial \tau_{nk}},\frac{\partial \bmu({\bgamma_{nk}})}{\partial f_{\rmd, nk}},\frac{\partial \bmu({\bgamma_{nk}})}{\partial \theta_{nk}}\bigg].\label{Dgamma}
\end{align}

Thereafter, the terms involved in $\bPsi(\bgamma_{nk})$ are given by
\begin{subequations}
	\begin{align}
		\frac{\partial \bmu({\bgamma_{nk}})}{\partial \widetilde \balpha_{nk}^\trans} &\!=\! [1, \jmath]\! \otimes \! [(\ba_{\rmt}(f_{\rmd, nk})\odot \bc_{nk}) \otimes \bs(\tau_{nk}) \otimes \ba_{\rms} (\theta_{nk})],\nonumber\\
		\frac{\partial \bmu({\bgamma_{nk}})}{\partial \tau_{nk}}&\!=\! \alpha_{nk} (\ba_{\rmt}(f_{\rmd, nk}) \! \odot \! \bc_{nk}) \! \otimes \!\frac{\partial  {\bs(\tau_{nk})}}{\partial \tau_{nk}} \!\otimes\! \ba_{\rms} (\theta_{nk}),\nonumber\\
		\disps\frac{\partial \bmu({\bgamma_{nk}})}{\partial f_{\rmd, nk}} &\!=\! \alpha_{nk} (\ba_{\rmt}(f_{\rmd, nk}) \! \odot \! \bc_{nk}\! \odot \! \bb_\rmt) \! \otimes \! {\bs(\tau_{nk})} \!\otimes\! \ba_{\rms} (\theta_{nk}), \nonumber\\
		\frac{\partial \bmu({\bgamma_{nk}})}{\partial \theta_{nk}} & \!= \!\alpha_{nk} (\ba_{\rmt}(f_{\rmd, nk}) \! \odot \! \bc_{nk} ) \! \otimes \! {\bs(\tau_{nk})} \otimes\! (\ba_{\rms} (\theta_{nk}) \!\odot \! \bb_{{\rms}}(\theta_{nk})),\nonumber
	\end{align}
\end{subequations}
where 
\begin{align}
	\!\!\!\!\bb_{\rms}(\theta) \!=\! \left[0, \jmath 2\pi d_n\frac{\cos\theta}{\lambda_n}, \cdots, \jmath 2\pi d_n \frac{\cos\theta}{\lambda_n}(N_\rmr\!-\!1)\right]^\trans\!\!\!\!.
\end{align}
Then, according to the Slepian-Bangs formula\cite{stoica1997intro}, the FIM  is
\beq
\mathbfcal{I}=\disps \left[ {\begin{array}{*{20}{l}}
		{\disps \mathbfcal{I}_{ \widetilde \balpha\widetilde \balpha}}&{\mathbfcal{I}_{\tau { \widetilde \balpha}}^\trans}&{\mathbfcal{I}_{{ f}{\widetilde \balpha}}^\trans}&{\mathbfcal{I}_{\theta { \widetilde \balpha}}^\trans}\\
		{\mathbfcal{I}_{\tau { \widetilde \balpha}}}&{\mathcal{I}_{\tau \tau}}&{\mathcal{I}_{f\tau  }^\trans}&{\mathcal{I}_{\theta \tau }^\trans}\\
		{\mathbfcal{I}_{f{ \widetilde \balpha}}}&{\mathcal{I}_{f \tau }}&{\mathcal{I}_{ff}}&{\mathcal{I}_{\theta f}^\trans}\\
		{\mathbfcal{I}_{\theta { \widetilde \balpha}}}&{\mathcal{I}_{\theta \tau}}&{\mathcal{I}_{ \theta f}}&{\mathcal{I}_{\theta \theta}}
\end{array}} \right],
\eeq
where 
\begin{subequations}
	\begin{align}
		\!\!\!\!\mathbfcal{I}_{\widetilde \balpha\widetilde\balpha} &= 2N_\rmp\ba^\ctrans_{\rms}
		(\theta_{nk})\bSigma_{{\rms}, nk}^{-1}\ba{_\rms}(\theta_{nk})\bc_{nk}^\ctrans\bM_{0,nk}\bc_{nk}\bI_2 \nonumber \\
		\mathbfcal{I}_{\tau \widetilde \balpha}&= 2\ba^\ctrans_{\rms}
		(\theta_{nk})\bSigma_{{\rms}, nk}^{-1}\ba{_\rms}(\theta_{nk})\nonumber\\
		&\quad\cdot \bc_{nk}^\ctrans\bM_{0,nk}\bc_{nk}\Re\left\{[1, \jmath]\! \otimes \! \alpha_{nk}^* \frac{\partial^\ctrans \! {\bs(\tau)}}{\partial \tau}{\bs(\tau)}\right\}\!\!, \nonumber \\
		\mathbfcal{I}_{f\widetilde\balpha}&=  2N_\rmp\ba^\ctrans_{\rms}
		(\theta_{nk})\bSigma_{{\rms}, nk}^{-1}\ba{_\rms}(\theta_{nk})\nonumber\\
		&\quad\cdot \Re\big\{[1, \jmath]\! \otimes \!\alpha_{nk}^*\bc_{nk}^\ctrans\bM_{1,nk}\bc_{nk} \big\}, \nonumber\\
		\mathbfcal{I}_{\theta\widetilde\balpha}&=  2N_\rmp\bc_{nk}^\ctrans\bM_{0,nk}\bc_{nk}\nonumber\\
		&\quad\cdot \Re\!\big\{[1, \jmath]\! \otimes \!\alpha_{nk}^* (\ba_{\rms}
		(\theta_{nk})\!\odot\!\bb_{\rms}(\theta_{nk}))^\ctrans\bSigma_{{\rms}, nk}^{-1}\ba{_\rms}(\theta_{nk})\!\big\}, \nonumber\\
		\mathcal{I}_{\tau \tau}&= 2|\alpha_{nk}|^2 \Big\|\frac{\partial  {\bs(\tau)}}{\partial \tau}\Big\|^2\ba^\ctrans_{\rms}
		(\theta_{nk})\bSigma_{{\rms}, nk}^{-1}\ba{_\rms}(\theta_{nk})\nonumber\\
		&\quad\cdot\bc_{nk}^\ctrans\bM_{0,nk}\bc_{nk}, \nonumber\\
		\mathcal{I}_{f \tau}&= 2|\alpha_{nk}|^2\ba^\ctrans_{\rms}
		(\theta_{nk})\bSigma_{{\rms}, nk}^{-1}\ba{_\rms}(\theta_{nk})\nonumber\\
		&\quad\cdot\Re\left\{\bc_{nk}^\ctrans\bM_{1,nk}^\ctrans\bc_{nk}  \frac{\partial^\ctrans\!  {\bs(\tau)}}{\partial \tau} {\bs(\tau)}\right\}\!\!, \nonumber\\
		\mathcal{I}_{\theta \tau}&= 2|\alpha_{nk}|^2 \bc_{nk}^\ctrans\bM_{0,nk}\bc_{nk}\nonumber\\
		&\quad\cdot\Re\!\left\{\!(\ba_{\rms}
		(\theta_{nk})\!\odot\!\bb_{\rms}(\theta_{nk}))^\ctrans\bSigma_{{\rms}, nk}^{-1} \ba_{\rms}
		(\theta_{nk}) {\bs^\ctrans\!(\tau)}\frac{\partial {\bs(\tau)}}{\partial \tau}\!\right\}\!\!, \nonumber\\
		\mathcal{I}_{ff}&= 2N_\rmp|\alpha_{nk}|^2 \ba^\ctrans_{\rms}
		(\theta_{nk})\bSigma_{{\rms}, nk}^{-1}\ba{_\rms}(\theta_{nk})  \bc_{nk}^\ctrans\bM_{2,nk}\bc_{nk}, \nonumber\\
		\mathcal{I}_{\theta f}&= 2N_\rmp|\alpha_{nk}|^2\nonumber\\
		&\quad\cdot\Re\left\{(\ba_{\rms}
		(\theta_{nk})\odot\bb_{\rms}(\theta_{nk}))^\ctrans\bSigma_{{\rms}, nk}^{-1} \ba_{\rms}
		(\theta_{nk})\right\}\!\!, \nonumber\\
		\mathcal{I}_{\theta \theta}&= 2N_\rmp|\alpha_{nk}|^2 \bc_{nk}^\ctrans\bM_{0,nk}\bc_{nk}\nonumber\\
		&\quad\cdot(\ba_{\rms}
		(\theta_{nk})\odot\bb_{\rms}(\theta_{nk}))^\ctrans\bSigma_{{\rms}, nk}^{-1} (\ba_{\rms}
		(\theta_{nk})\odot\bb_{\rms}(\theta_{nk})). \nonumber 
	\end{align}
\end{subequations}

Assuming $\tau_0 \approx \tau$, then $\Big\|\frac{\partial  {\bs(\tau)}}{\partial \tau}\Big\|^2=\disps\frac{1}{3}\pi^2B^2N_\rmp$ and $\frac{\partial^\ctrans\!  {\bs(\tau)}}{\partial \tau}{\bs(\tau)}=0$\cite{dogandzic2001cramer}. Further, one has $\mathbfcal{I}_{\tau \widetilde \balpha}=\bzero_2$, $\mathcal{I}_{f \tau}=0$, $\mathcal{I}_{\theta \tau}=0$, and 
\begin{align}
	\mathcal{I}_{\tau \tau}= \disps \frac{2N_\rmp|\alpha_{nk}|^2\pi^2B^2}{3}\ba^\ctrans_{\rms}
	(\theta_{nk})\bSigma_{{\rms}, nk}^{-1}\ba{_\rms}(\theta_{nk})\bc_{nk}^\ctrans\bM_{0,nk}\bc_{nk}.
\end{align}

As a consequence, the CRLB for the time-delay, Doppler shift, and arrival angle is 
\beq \label{CRLBft}
\begin{aligned}
	\textbf{CRLB}_{nk}(\bc_{nk})&=\left\{\left[ {\begin{array}{*{20}{l}}
			{\mathcal{I}_{\tau \tau}}&{\mathcal{I}_{f \tau }^\trans}&{\mathcal{I}_{\theta \tau}^\trans}\\
			{\mathcal{I}_{f \tau }}&{\mathcal{I}_{ff}}&{\mathcal{I}_{\theta f }^\trans}\\
			{\mathcal{I}_{\theta \tau}}&{\mathcal{I}_{\theta f}}&{\mathcal{I}_{\theta \theta}}
	\end{array}} \right]- \mathbfcal{J}\mathbfcal{I}_{\widetilde\balpha\widetilde\balpha}^{-1}\mathbfcal{J}^\trans\right\}^{\!\!-1}\\
	&=\left[ {\begin{array}{*{20}{l}}
			{\mathcal{\tilde I}_{\tau \tau}}&{\mathcal{\tilde I}_{f \tau }^\trans}&{\mathcal{\tilde I}_{\theta \tau}^\trans}\\
			{\mathcal{\tilde I}_{f \tau }}&{\mathcal{\tilde I}_{ff}}&{\mathcal{\tilde I}_{\theta f }^\trans}\\
			{\mathcal{\tilde I}_{\theta \tau}}&	{\mathcal{\tilde I}_{\theta f}}&{\mathcal{\tilde I}_{\theta \theta}}
	\end{array}} \right]^{\!\!-1},
\end{aligned}
\eeq
where $\mathbfcal{J}=[\mathbfcal{I}_{\tau\widetilde\balpha}^\trans,\mathbfcal{I}_{f\widetilde\balpha}^\trans, \mathbfcal{I}_{\theta \widetilde\balpha}^\trans]^\trans$. It follows that 
\beq\label{JIJ}
\begin{aligned}
	\!\!\!\!\mathbfcal{J}&\mathbfcal{I}_{\widetilde\balpha\widetilde\balpha}^{-1}\mathbfcal{J}^\trans\\
	=&\frac{1}{2N_\rmp\ba^\ctrans_{\rms}
		(\theta_{nk})\bSigma_{{\rms}, nk}^{-1}\ba{_\rms}(\theta_{nk})\bc_{nk}^\ctrans\bM_{0,nk}\bc_{nk}}\\
	&\cdot \left[\!\! {\begin{array}{*{20}{l}}
			{\Re\{P_{\tau\widetilde\balpha}P_{\tau\widetilde\balpha}^*\}}&	{\Re\{P_{\tau\widetilde\balpha}P_{f\widetilde\balpha}^*\}}&{\Re\{P_{\tau\widetilde\balpha}P_{\theta\widetilde\balpha}^*\}}\\
			{\Re\{P_{f\widetilde\balpha}P_{\tau \widetilde\balpha}^*\}}&{\Re\{P_{f\widetilde\balpha}P_{f\widetilde\balpha}^*\}}&{\Re\{P_{f\widetilde\balpha}P_{\theta \widetilde\balpha}^*\}}\\
			{\Re\{P_{\theta\widetilde\balpha}P_{\tau\widetilde\balpha}^*\}}&{\Re\{P_{\theta\widetilde\balpha}P_{f\widetilde\balpha}^*\}}&{\Re\{P_{\theta\widetilde\balpha}P_{\theta\widetilde\balpha}^*\}}
	\end{array}} \!\!\right],
\end{aligned}
\eeq
where 
\begin{subequations}
	\begin{align}
		P_{\tau \widetilde\balpha} =& 0, \nonumber \\
		P_{f\widetilde\balpha} =& 2N_\rmp\alpha_{nk}^*\ba^\ctrans_{\rms}
		(\theta_{nk})\bSigma_{{\rms}, nk}^{-1}\ba{_\rms}(\theta_{nk})\bc_{nk}^\ctrans\bM_{1,nk}\bc_{nk},  \nonumber \\
		P_{\theta \widetilde\balpha} =& 2N_\rmp\alpha_{nk}^*(\ba_{\rms}
		(\theta_{nk})\odot\bb_{\rms}(\theta_{nk}))^\ctrans\bSigma_{{\rms}, nk}^{-1}\ba{_\rms}(\theta_{nk})\nonumber\\
		&\times\bc_{nk}^\ctrans\bM_{0,nk}\bc_{nk}. \nonumber
	\end{align}
\end{subequations}
Substituting equation \eqref{JIJ} into equation \eqref{CRLBft}, yields
\begin{subequations}\label{Itf}
	\begin{align}
		\mathcal{\tilde I}_{\tau \tau} &= \mathcal{I}_{\tau \tau}=  \varepsilon_{\tau, nk}\bc_{nk}^\ctrans\bM_{0,nk}\bc_{nk},\\
		\mathcal{\tilde I}_{\tau f}&= 0,   \\
		\mathcal{\tilde I}_{\theta \tau} &= 0,\\
		\mathcal{\tilde I}_{ff} &= \varepsilon_{f, nk}\!\!\left({\bc_{nk}^\ctrans}\bM_{2,nk}\bc_{nk}\!-\!\frac{\big|{\bc_{nk}^\ctrans}\bM_{1,nk}\bc_{nk}\big|^2}{{\bc_{nk}^\ctrans}\bM_{0,nk}\bc_{nk}}\!\right)\!\!,\label{Iff1}\\
		\mathcal{\tilde I}_{\theta f}&= 0,   \\
		\mathcal{\tilde I}_{\theta \theta} &= \varepsilon_{\theta, nk} \bc_{nk}^\ctrans\bM_{0,nk}\bc_{nk},\\
	\end{align}
\end{subequations}
where
\begin{subequations}\label{epss_nk}
	\begin{align}
		\varepsilon_{\tau, nk} &\!=\!   \disps \frac{2N_\rmp|\alpha_{nk}|^2\pi^2B^2}{3}\ba^\ctrans_{\rms}
		(\theta_{nk})\bSigma_{{\rms}, nk}^{-1}\ba{_\rms}(\theta_{nk}), \label{epst_nk}\\
		\varepsilon_{f, nk} &\!=\!2N_\rmp|\alpha_{nk}|^2  \ba^\ctrans_{\rms}
		(\theta_{nk})\bSigma_{{\rms}, nk}^{-1}\ba{_\rms}(\theta_{nk}), \label{epsf_nk}\\
		\varepsilon_{\theta, nk} &\!=\!  2N_\rmp|\alpha_{nk}|^2\nonumber\\
		&\quad \cdot\!\! \Bigg(\!(\ba_{\rms}
		(\theta_{nk})\!\odot\!\bb_{\rms}(\theta_{nk}))^\ctrans\bSigma_{{\rms}, nk}^{-1} (\ba_{\rms}
		(\theta_{nk})\!\odot\!\bb_{\rms}(\theta_{nk}))\nonumber\\
		&\quad\quad\quad\!-\! \disps \frac{\big|(\ba_{\rms}
			(\theta_{nk})\!\odot\!\bb_{\rms}(\theta_{nk}))^\ctrans\bSigma_{{\rms}, nk}^{-1}\ba{_\rms}(\theta_{nk})\big|^2}{\ba^\ctrans_{\rms}
			(\theta_{nk})\bSigma_{{\rms}, nk}^{-1}\ba{_\rms}(\theta_{nk})}\!\Bigg)\!.\label{epsa_nk}
	\end{align}
\end{subequations}
\subsection{Proof of Lemma \ref{prop2}}\label{proof2ndexpansion}
Before presenting the proof, the following lemmas that will be used are provided.
\begin{lemma}\label{lemmaQfunc}
	For a fixed $b$, the gradient and the Hessian matrix of the generalized Marcum Q function $Q_v(\sqrt{2x}, b)$ are, respectively, expressed as
	\begin{align}
		\frac{\partial Q_v(\sqrt{2x}, b)}{\partial x} =& Q_{v+1}(\sqrt{2x}, b) - Q_{v}(\sqrt{2x}, b),\\
		\frac{\partial^2 Q_v(\sqrt{2x}, b)}{\partial x\partial x^\trans} =& Q_{v+2}(\sqrt{2x}, b)-2Q_{v+1}(\sqrt{2x}, b)\nonumber \\
		&+Q_{v}(\sqrt{2x}, b).
	\end{align}
	\begin{proof}
		See \cite{segura2014monotonicity} for details.
	\end{proof}
\end{lemma}
\begin{lemma}\label{lemmaphifunc}
	The gradient and the Hessian matrix of $\widetilde\phi(\widetilde \bc_{nk})$ are, respectively, given by
	\begin{align}
		&\bphi(\widetilde \bc_{nk})=\frac{\partial \widetilde\phi(\widetilde \bc_{nk})}{\partial \widetilde \bc_{nk}}\nonumber\\
		&\;= 2 \widetilde\bc_{nk}^\trans \!\widetilde\bM_{2, nk} \widetilde\bc_{nk}\!\widetilde\bM_{0, nk} \widetilde\bc_{nk}+2\widetilde\bc_{nk}^\trans \!\widetilde\bM_{0, nk} \widetilde\bc_{nk}\widetilde\bM_{2, nk} \widetilde\bc_{nk} \nonumber\\
		&\;\quad-2\widetilde\bc_{nk}^\trans \widetilde\bM_{1,nk} \widetilde\bc_{nk} (\widetilde\bM_{1,nk}+\widetilde\bM_{1,nk}^\trans) \widetilde\bc_{nk}\nonumber\\
		&\;\quad-2\widetilde\bc_{nk}^\trans \widehat\bM_{1,nk} \widetilde\bc_{nk} (\widehat\bM_{1,nk}+\widehat\bM_{1,nk}^\trans) \widetilde\bc_{nk}\\
		&\bPhi(\widetilde \bc_{nk})=\frac{\partial^2 \widetilde\phi(\widetilde \bc_{nk})}{\partial \widetilde \bc_{nk}\partial {\widetilde \bc_{nk}}^\trans}
		\nonumber\\
		&\;= 2 \widetilde\bc_{nk}^\trans \!\widetilde\bM_{2, nk} \widetilde\bc_{nk}\!\widetilde\bM_{0, nk} \!+\!4  \widetilde\bM_{0, nk} \widetilde\bc_{nk}\widetilde\bc_{nk}^\trans\widetilde\bM_{2, nk}^\trans \nonumber\\
		&\;\quad+2\widetilde\bc_{nk}^\trans \!\widetilde\bM_{0, nk} \widetilde\bc_{nk}\widetilde\bM_{2, nk}\!+\!4\widetilde\bM_{2, nk} \widetilde\bc_{nk}\widetilde\bc_{nk}^\trans\widetilde\bM_{0, nk}^\trans  \nonumber\\
		&\;\quad-2\widetilde\bc_{nk}^\trans \widetilde\bM_{1,nk} \widetilde\bc_{nk} (\widetilde\bM_{1,nk}+\widetilde\bM_{1,nk}^\trans) \nonumber\\
		&\;\quad-2   (\widetilde\bM_{1,nk}+\widetilde\bM_{1,nk}^\trans)\widetilde\bc_{nk}\widetilde\bc_{nk}^\trans(\widetilde\bM_{1,nk}+\widetilde\bM_{1,nk}^\trans)^\trans \nonumber\\
		&\;\quad-2\widetilde\bc_{nk}^\trans \widehat\bM_{1,nk} \widetilde\bc_{nk} (\widehat\bM_{1,nk}+\widehat\bM_{1,nk}^\trans)\nonumber\\
		&\;\quad-2 (\widehat\bM_{1,nk}+\widehat\bM_{1,nk}^\trans) \widetilde\bc_{nk}\widetilde\bc_{nk}^\trans(\widehat\bM_{1,nk}+\widehat\bM_{1,nk}^\trans)^\trans.
	\end{align}
\end{lemma}

According to Lemma \ref{lemmaQfunc}, the gradient and the Hessian matrix of $\widetilde P_{\rmd, nk}(\widetilde \bc_{nk})$ are, respectively, represented\footnote{With a slight abuse of notation, $Q_{v}\Big(\sqrt{2 \widetilde {\text{SINR}}(\widetilde \bc_{nk},\tau_{nk},  f_{\rmd, nk}, \theta_{nk})}, \sqrt{2b_0}\Big)$ is denoted by $Q_{v}\big(\widetilde {\text{SINR}}(\widetilde \bc_{nk})\big)$ for $v=1,2,3$.} as
\begin{align}
	&\bp(\widetilde \bc_{nk})=\frac{\partial \widetilde P_{\rmd, nk}(\widetilde \bc_{nk})}{\widetilde \bc_{nk}} \nonumber\\
	&=\! 2\varepsilon_{\alpha, nk} \widetilde\bM_{0, nk}\widetilde\bc_{nk}\!\Big[Q_{2}\big(\widetilde {\text{SINR}}(\widetilde \bc_{nk})\big)\!-\!Q_{1}\!\big(\widetilde {\text{SINR}}(\widetilde \bc_{nk})\big)\!\Big]\!,\\
	&\bP(\widetilde \bc_{nk})=\frac{\partial^2 \widetilde P_{\rmd, nk}(\widetilde \bc_{nk})}{\partial \widetilde \bc_{nk}\partial \widetilde \bc_{nk}^\trans} \nonumber\\
	&= 2\varepsilon_{\alpha, nk} \widetilde\bM_{0, nk}\Big[Q_{2}\big(\widetilde {\text{SINR}}(\widetilde \bc_{nk})\big)\!-\!Q_{1}\!\big(\widetilde {\text{SINR}}(\widetilde \bc_{nk})\big)\!\Big]\!\nonumber\\
	&\quad+\!4\varepsilon_{\alpha, nk}^2 \widetilde\bM_{0, nk}\widetilde\bc_{nk}\widetilde\bc_{nk}^\trans\widetilde\bM_{0, nk}^\trans\Big[Q_{3}\big(\widetilde {\text{SINR}}(\widetilde \bc_{nk})\big)\nonumber\\
	&\quad\quad-\!2Q_{2}\big(\widetilde {\text{SINR}}(\widetilde \bc_{nk})\big)\!+\!Q_{1}\big(\widetilde {\text{SINR}}(\widetilde \bc_{nk})\big)\Big].
\end{align}

Based on the above discussion, the approximated value of $({\widetilde P_{\rmd, nk}(\widetilde \bc_{nk})\widetilde\bc_{nk}^\trans \widetilde\bM_{0, nk} \widetilde\bc_{nk}})^{-1}$ in the neighborhood of $\widetilde \bc_{0}$ can be written as
\begin{align}
	&\big(\widetilde P_{\rmd, nk}(\widetilde \bc_{nk})\widetilde\bc_{nk}^\trans \widetilde\bM_{0, nk} \widetilde\bc_{nk}\big)^{-1}\nonumber\\
	&\approx\big(\widetilde P_{\rmd, nk}(\widetilde\bc_{0})\widetilde\bc_{0}^\trans \widetilde\bM_{0, nk} \widetilde\bc_{0}\big)^{-1}\nonumber\\
	&\quad+\left.\bigg(\frac{\partial \big(\widetilde P_{\rmd, nk}(\widetilde\bc_{nk})\widetilde\bc_{nk}^\trans \widetilde\bM_{0, nk} \widetilde\bc_{nk}\big)^{-1}}{\partial\widetilde \bc_{nk}} \bigg)^\trans\right|_{\widetilde \bc_{nk}=\widetilde\bc_{0}}\!\!\!\!\!\!\!\!\!\!\!\!\!\!\!\!(\widetilde \bc_{nk}-\widetilde\bc_{0})\nonumber\\
	&\quad +\disps\frac{1}{2}(\widetilde \bc_{nk}-\widetilde\bc_{0})^\trans\bGamma_1(\widetilde \bc_{0})(\widetilde \bc_{nk}-\widetilde\bc_{0}),
\end{align}
where
\begin{align}
	\!\!\!\!\!\!\!\!&\frac{\partial \big(\widetilde P_{\rmd, nk}(\widetilde\bc_{0})\widetilde\bc_{0}^\trans \widetilde\bM_{0, nk} \widetilde\bc_{0}\big)^{-1}}{\partial\widetilde \bc_{nk}}\nonumber\\
	&\;=\disps-\frac{\bp(\widetilde \bc_{nk})\widetilde\bc_{nk}^\trans \widetilde\bM_{0, nk} \widetilde\bc_{nk}+2\widetilde P_{\rmd, nk}(\widetilde \bc_{nk}) \widetilde\bM_{0, nk} \widetilde\bc_{nk}}{(\widetilde P_{\rmd, nk}(\widetilde \bc_{nk})\widetilde\bc_{nk}^\trans \widetilde\bM_{0, nk} \widetilde\bc_{nk})^2},
\end{align}
and 
\begin{align}
	\!\!\!\!\!\!\bGamma_1(\widetilde \bc_{0})&=\!\!\left.\frac{\partial^2 \big(\widetilde P_{\rmd, nk}(\widetilde\bc_{nk})\widetilde\bc_{nk}^\trans \widetilde\bM_{0, nk} \widetilde\bc_{nk}\big)^{-1}}{\partial\widetilde \bc_{nk}\partial\widetilde \bc_{nk}^\trans}\right|_{\widetilde \bc_{nk}=\widetilde\bc_{0}}\nonumber\\
	&\;=\!\!-\frac{\bP(\widetilde \bc_{0})\widetilde\bc_{0}^\trans \widetilde\bM_{0, nk} \widetilde\bc_{0}+2\bp(\widetilde \bc_{0})\widetilde\bc_{0}^\trans \widetilde\bM_{0, nk}^\trans}{(\widetilde P_{\rmd, nk}(\widetilde \bc_{0})\widetilde\bc_{0}^\trans \widetilde\bM_{0, nk} \widetilde\bc_{0})^2}\nonumber\\
	&\;\quad\!\!-\frac{2 \widetilde\bM_{0, nk} \widetilde\bc_{0}\bp(\widetilde \bc_{0})^\trans+2\widetilde P_{\rmd, nk}(\widetilde \bc_{0}) \widetilde\bM_{0, nk}}{(\widetilde P_{\rmd, nk}(\widetilde \bc_{0})\widetilde\bc_{0}^\trans \widetilde\bM_{0, nk} \widetilde\bc_{0})^2}\nonumber\\
	&\;\quad\!\!+\disps\frac{2\big[\bp(\widetilde \bc_{0})\widetilde\bc_{0}^\trans \widetilde\bM_{0, nk} \widetilde\bc_{0}\!+\!2\widetilde P_{\rmd, nk}(\widetilde \bc_{0}) \widetilde\bM_{0, nk} \widetilde\bc_{0}\big]}{(\widetilde P_{\rmd, nk}(\widetilde \bc_{0})\widetilde\bc_{0}^\trans \widetilde\bM_{0, nk} \widetilde\bc_{0})^3}\nonumber\\
	&\;\quad\quad\!\!\cdot \big[\bp(\widetilde \bc_{0})\widetilde\bc_{0}^\trans \widetilde\bM_{0, nk} \widetilde\bc_{0}\!+\!2\widetilde P_{\rmd, nk}(\widetilde \bc_{0}) \widetilde\bM_{0, nk} \widetilde\bc_{0}\big]^\trans.\!\!\!\!
\end{align}

Therefore, the approximated valued of $[\bS_{nk}(\widetilde\bc_{nk})]_{11}$ and $[\widetilde \bR_{nk}(\widetilde\bc_{nk})]_{33}$ are, respectively, given by
\begin{align}
	[\widehat \bS_{nk}(\widetilde\bc_{nk})]_{11} \approx & \widetilde \bc_{nk}^\trans\bA_{1, nk} \widetilde \bc_{nk} + \ba_{1, nk}^\trans \widetilde \bc_{nk} + a_{1, nk},\\
	[\widehat \bS_{nk}(\widetilde\bc_{nk})]_{33} \approx & \widetilde \bc_{nk}^\trans\bA_{3, nk} \widetilde \bc_{nk} + \ba_{3, nk}^\trans \widetilde \bc_{nk} + a_{3, nk},
\end{align}
where 
\begin{itemize}
	\item $\bA_{1, nk}=\disps\frac{c^2}{8\varepsilon_{\tau, nk}}\bGamma_1(\widetilde \bc_{0})$;
	\item $\ba_{1, nk}\!=\!\disps-\frac{c^2}{4\varepsilon_{\tau, nk}}\Big(\frac{\bp(\widetilde \bc_{0})\widetilde\bc_{0}^\trans \widetilde\bM_{0, nk} \widetilde\bc_{0}}{(\widetilde P_{\rmd, nk}(\widetilde \bc_{0})\widetilde\bc_{0}^\trans \widetilde\bM_{0, nk} \widetilde\bc_{0})^2}+\frac{2\widetilde P_{\rmd, nk}(\widetilde \bc_{0}) \widetilde\bM_{0, nk} \widetilde\bc_{0}}{(\widetilde P_{\rmd, nk}(\widetilde \bc_{0})\widetilde\bc_{0}^\trans \widetilde\bM_{0, nk} \widetilde\bc_{0})^2}+\bGamma_1(\widetilde \bc_{0})^\trans\widetilde\bc_{0}\Big)$;
	\item $a_{1, nk}=\disps\frac{c^2}{4\varepsilon_{\tau, nk}}\bigg(\disps\frac{1}{\widetilde P_{\rmd, nk}(\widetilde\bc_{0})\widetilde\bc_{0}^\trans \widetilde\bM_{0, nk} \widetilde\bc_{0}}+\frac{\bp(\widetilde \bc_{0})^\trans\widetilde\bc_{0}\widetilde\bc_{0}^\trans \widetilde\bM_{0, nk} \widetilde\bc_{0}+2\widetilde P_{\rmd, nk}(\widetilde \bc_{0})\widetilde\bc_{0}^\trans \widetilde\bM_{0, nk} \widetilde\bc_{0}}{(\widetilde P_{\rmd, nk}(\widetilde \bc_{0})\widetilde\bc_{0}^\trans \widetilde\bM_{0, nk} \widetilde\bc_{0})^2}+\frac{1}{2}\widetilde\bc_{0}^\trans\bGamma_1(\bc_{0})\widetilde\bc_{0}\bigg)$;
	\item $\bA_{3, nk}=\disps\frac{4\varepsilon_{\tau, nk}}{c^2\varepsilon_{\theta, nk}}\bA_{1, nk}$;
	\item $\ba_{3, nk}=\disps\frac{4\varepsilon_{\tau, nk}}{c^2\varepsilon_{\theta, nk}}\ba_{1, nk}$;
	\item $a_{3, nk}=\disps\frac{4\varepsilon_{\tau, nk}}{c^2\varepsilon_{\theta, nk}}a_{1, nk}$.
\end{itemize}

Now, based on Lemmas \ref{lemmaQfunc} and \ref{lemmaphifunc}, the term $\frac{\widetilde\bc_{nk}^\trans \widetilde\bM_{0, nk} \widetilde\bc_{nk}}{\widetilde P_{\rmd, nk}(\widetilde \bc_{nk})\phi(\widetilde \bc_{nk})}$ can be approximated by
\begin{align}
	&\disps\frac{\widetilde\bc_{nk}^\trans \widetilde\bM_{0, nk} \widetilde\bc_{nk}}{\widetilde P_{\rmd, nk}(\widetilde \bc_{nk})\phi(\widetilde \bc_{nk})}\quad\approx\disps\frac{\widetilde\bc_{0}^\trans \widetilde\bM_{0, nk} \widetilde\bc_{0}}{\widetilde P_{\rmd, nk}(\widetilde\bc_{0})\phi(\widetilde \bc_{0})}\nonumber\\
	&\quad\quad+\left.\bigg(\frac{\partial}{\partial \bc_{nk}} \bigg[\frac{\widetilde\bc_{nk}^\trans \widetilde\bM_{0, nk} \widetilde\bc_{nk}}{\widetilde P_{\rmd, nk}(\widetilde \bc_{nk})\phi(\widetilde \bc_{nk})}\bigg]\bigg)^\trans\right|_{\widetilde \bc_{nk}=\widetilde\bc_{0}}\!\!\!\!\!\!\!\!\!\!\!\!\!\!(\widetilde \bc_{nk}-\widetilde\bc_{0})\nonumber\\
	&\quad\quad +\disps\frac{1}{2}(\widetilde \bc_{nk}-\widetilde\bc_{0})^\trans\bGamma_2(\widetilde \bc_{0})(\widetilde \bc_{nk}-\widetilde\bc_{0}),
\end{align}
where 
\begin{align}
	&\left.\frac{\partial}{\partial \bc_{nk}} \bigg[\frac{\widetilde\bc_{nk}^\trans \widetilde\bM_{0, nk} \widetilde\bc_{nk}}{\widetilde P_{\rmd, nk}(\widetilde \bc_{nk})\phi(\widetilde \bc_{nk})}\bigg]\right|_{\widetilde\bc_{nk}=\widetilde\bc_0}\nonumber\\
	&\;=\disps\frac{2 \widetilde\bM_{0, nk} \widetilde\bc_{nk}}{\widetilde P_{\rmd, nk}(\widetilde \bc_{nk})\phi(\widetilde \bc_{nk})}-\disps\frac{\widetilde\bc_{nk}^\trans \widetilde\bM_{0, nk} \widetilde\bc_{nk}\bq(\widetilde \bc_{nk})}{(\widetilde P_{\rmd, nk}(\widetilde \bc_{nk})\phi(\widetilde \bc_{nk}))^2}
\end{align}
with $\bq(\widetilde \bc_{nk})=\bp({\widetilde \bc_{nk}})\phi(\widetilde \bc_{nk})+\widetilde P_{\rmd, nk}(\widetilde \bc_{nk})\bphi(\widetilde \bc_{nk})$, and 
\begin{align}
	\bGamma_2(\widetilde \bc_{nk})
	&\;=\left.\frac{\partial^2}{\partial \bc_{nk}\partial \bc_{nk}^\trans} \bigg[\frac{\widetilde\bc_{nk}^\trans \widetilde\bM_{0, nk} \widetilde\bc_{nk}}{\widetilde P_{\rmd, nk}(\widetilde \bc_{nk})\phi(\widetilde \bc_{nk})}\bigg]\right|_{\widetilde\bc_{nk}=\widetilde\bc_0}\nonumber\\
	&\;=\disps\frac{2 \widetilde\bM_{0, nk}}{\widetilde P_{\rmd, nk}(\widetilde \bc_{0})\phi(\widetilde \bc_{0})}-\disps\frac{2 \widetilde\bM_{0, nk} \widetilde\bc_{0}\bq(\widetilde \bc_{0})^\trans}{(\widetilde P_{\rmd, nk}(\widetilde \bc_{0})\phi(\widetilde \bc_{0}))^2}\nonumber\\
	&\;\quad-\!\disps\frac{2 \bq(\widetilde \bc_{0})\widetilde\bc_{0}^\trans \widetilde\bM_{0, nk}^\trans\!+\!\widetilde\bc_{0}^\trans \widetilde\bM_{0, nk} \widetilde\bc_{0}\bQ(\widetilde\bc_{0})}{(\widetilde P_{\rmd, nk}(\widetilde \bc_{0})\phi(\widetilde \bc_{0}))^2}\nonumber\\
	&\;\quad+\!\disps\frac{2\widetilde\bc_{0}^\trans \widetilde\bM_{0, nk} \widetilde\bc_{0}\bq(\widetilde \bc_{0})\bq(\widetilde \bc_{0})^\trans}{(\widetilde P_{\rmd, nk}(\widetilde \bc_{0})\phi(\widetilde \bc_{0}))^3}.
\end{align}
with $\bQ(\widetilde\bc_{nk})=\bP(\widetilde \bc_{nk})\phi(\widetilde \bc_{nk})\!+\!\bp({\widetilde \bc_{nk}})\bphi(\widetilde \bc_{nk})^\trans\!+\!\bphi(\widetilde \bc_{nk})\bp(\widetilde\bc_{nk})^\trans\!+\!\widetilde P_{\rmd, nk}(\widetilde\bc_{nk})\bPhi(\widetilde \bc_{nk})$. As a consequence, $[\widetilde \bR_{nk}(\widetilde\bc_{nk})]_{22}$ can be approximated by
\begin{align}
	[\widehat \bS_{nk}(\widetilde\bc_{nk})]_{22} \approx & \widetilde \bc_{nk}^\trans\bA_{2, nk} \widetilde \bc_{nk} + \ba_{2, nk}^\trans \widetilde \bc_{nk} + a_{2, nk},
\end{align}
\begin{itemize}
	\item $\bA_{2, nk}=\disps\frac{\lambda_n^2}{8\varepsilon_{f_\rmd, nk}}\bGamma_2(\widetilde \bc_{0})$;
	\item $\ba_{2, nk}=\disps-\frac{\lambda_n^2}{4\varepsilon_{f_\rmd, nk}}\Big(\bGamma_2(\widetilde \bc_{0})^\trans\widetilde\bc_{0}-\disps\frac{2 \widetilde\bM_{0, nk} \widetilde\bc_{0}}{\widetilde P_{\rmd, nk}(\widetilde \bc_{0})\phi(\widetilde \bc_{0})}+\disps\frac{\widetilde\bc_{0}^\trans \widetilde\bM_{0, nk} \widetilde\bc_{0}\bq(\widetilde \bc_{0})}{(\widetilde P_{\rmd, nk}(\widetilde \bc_{0})\phi(\widetilde \bc_{0}))^2}\Big)$;
	\item $a_{2, nk}=\disps\frac{\lambda_n^2}{4\varepsilon_{f_\rmd, nk}}\bigg(-\disps\frac{ \widetilde \bc_{0}^\trans\widetilde\bM_{0, nk} \widetilde\bc_{0}}{\widetilde P_{\rmd, nk}(\widetilde \bc_{0})\phi(\widetilde \bc_{0})}+\disps\frac{\widetilde\bc_{0}^\trans \widetilde\bM_{0, nk} \widetilde\bc_{0}\bq(\widetilde \bc_{0})^\trans\widetilde \bc_{0}}{(\widetilde P_{\rmd, nk}(\widetilde \bc_{0})\phi(\widetilde \bc_{0}))^2}+\frac{1}{2}\widetilde\bc_{0}^\trans\bGamma_2(\widetilde \bc_{0})\widetilde\bc_{0}\bigg)$.
\end{itemize}

\subsection{Proof of Proposition \ref{prop_fRconcave}}\label{proof_prop_fRconcave}
A corollary for the concavity analysis of a real-valued function is provided \cite{boyd2004convex}. 
\begin{corollary}\label{corollary1}
	{\rm Given a real-valued function $ g\!: \bbR^{M\times M}\!\rightarrow\!\bbR$, $\bX,\bY\!\in\! \dom(g)$, define the restriction $\psi_{(\bX,\bY)}\!:\bbR\rightarrow\bbR$ as $\psi_{(\bX,\bY)}(t)=g(\bX+t\bY)$. Then, $f$ is convex on $\bbR^{M\times M}$ if and only if for any $\bX, \bY \in \dom(f)$, $\psi_{(\bX,\bY)}(t)$ is concave on $\bbR$; $\psi''_{(\bX,\bY)}(t)\leq 0$ for all $t\in \bbR$, assuming that the second order derivatives exists.} 
\end{corollary}

Based on the Woodbury matrix identity\cite{higham2002accuracy}, $g(\bR)$ can be recast as 
\begin{align}\label{fR1}
	\!g(\bR) \!=&\Tr\!(\bB) \!-\! \Tr\!\Big(\bB\bH^\trans(\bR_{}\!+\!\bH\bB\bH^\trans)^{-1} \bH\bB\Big).
\end{align}
According to Corollary \ref{corollary1}, to prove that $ g(\bR)$ is concave, it is necessary to show that $\psi_{(\bZ,\bV)}(t)\!= g(\bZ+t\bV)$ for any $\bZ, \bV \in \bbS^{M}_{++}$ is concave. Letting $\widehat \bZ=\bZ+\bH\bB\bH^\trans$ which implies that $\widehat \bZ\in \bbS_{++}^{M}$ because $\bB_{}\in \bbS^{M}_{++}$, then $\psi_{(\bZ,\bV)}(t)$ can be expressed as
\begin{align}
	&\psi_{(\bZ,\bV)}(t) \nonumber\\
	&=\! \Tr(\bB) \!-\! \Tr\Big(\bB\bH^\trans(\widehat\bZ\!+\!t\bV)^{-1} \bH\bB\Big)\nonumber\\
	&=\! \Tr(\bB) \!-\! \Tr\Big(\bB\bH^\trans\widehat\bZ^{-\frac{1}{2}}(\bI\!+\!t\widehat\bZ^{-\frac{1}{2}}\bV\widehat\bZ^{-\frac{1}{2}})^{-1}\widehat\bZ^{-\frac{1}{2}} \bH\bB\Big)\nonumber\\
	&=\! \Tr(\bB) \!-\! \Tr\Big(\widehat\bZ^{-\frac{1}{2}} \bH\bB\bB\bH^\trans\widehat\bZ^{-\frac{1}{2}}(\bI\!+\!t\widehat\bZ^{-\frac{1}{2}}\bV\widehat\bZ^{-\frac{1}{2}})^{-1}\Big).
\end{align}

Since $\widehat \bZ,\bV\in \bbS_{++}^{M}$, $\widehat\bZ^{-\frac{1}{2}}\bV\widehat\bZ^{-\frac{1}{2}}$ can be decomposed as $\widehat\bZ^{-\frac{1}{2}}\bV\widehat\bZ^{-\frac{1}{2}}=\bU\bLambda\bU^\trans$, where $\bU$ is an orthogonal matrix and $\bLambda$ is a diagonal matrix whose diagonal entries are the eigenvalues of $\widehat\bZ^{-\frac{1}{2}}\bV\widehat\bZ^{-\frac{1}{2}}$ and are both greater than zero. As a consequence, $\psi_{(\bZ,\bV)}(t)$ can be further expressed as
\begin{align}
	&\psi_{(\bZ,\bV)}(t)\nonumber\\
	& =\! \Tr(\bB) \!-\! \Tr\Big(\widehat\bZ^{-\frac{1}{2}} \bH\bB\bB\bH^\trans\widehat\bZ^{-\frac{1}{2}}(\bU(\bI\!+\!t\bLambda)\bU^\trans)^{-1}\Big)\nonumber\\
	& =\! \Tr(\bB) \!-\! \Tr\Big(\bU\widehat\bZ^{-\frac{1}{2}} \bH\bB\bB(\bU\widehat\bZ^{-\frac{1}{2}} \bH)^\trans(\bI\!+\!t\bLambda)^{-1}\Big)\nonumber\\
	& =\!\Tr(\bB) \!-\!\Tr\Big(\bC(\bI\!+\!t\bLambda)^{-1}\Big),
\end{align}
where $\bC=\bU\widehat\bZ^{-\frac{1}{2}} \bH\bB\bB(\bU\widehat\bZ^{-\frac{1}{2}} \bH)^\trans$ is a symmetric and positive definite matrix, because $\bU\widehat\bZ^{-\frac{1}{2}} \bH$ is not zero matrix and $\bB\in\bbS_{++}^M$. 
Therefore, the second order derivative of $\psi_{(\bZ,\bV)}(t)$ is
\begin{align}
	\psi''_{(\bZ,\bV)}(t) \!=\!& -\sum_{m=1}^{M}\bC(m,m)\disps\frac{2(\bLambda(m,m))^2}{(1+t\bLambda(m,m))^3}.
\end{align}
Because of $\bC(m,m)>0$ and $\bLambda(m,m)>0$ for $m=1,2,\cdots,M$, it results $\psi''_{(\bZ,\bV)}(t)<0$, indicating that  $g(\bR)$ is a concave function for $\bR\in \bbS_{++}^M$.
\subsection{Proof of Lemma \ref{lemma_Pk_sca_realx}}\label{proof_Pk_sca_relax}
Note that Problem \eqref{Pk_sca_relax} is a convex optimization problem with a non-empty and compact set, thus, it has at least one optimal solution. The main idea of the proof is to establish the existence of an optimal solution to the relaxed Problem \eqref{Pk_sca_relax} that satisfies $\|\widetilde \bc_{pk}\|^2\!=\!1$ in Problem \eqref{Pk_sca}. To verify this, it is enough to show that, for any feasible point $\bar \bc_{pk}$ with $\|\bar \bc_{pk}\|^2 < 1$ of Problem \eqref{Pk_sca_relax}, one can always find a new feasible point with unit energy that yields the lower or same objective value as the original feasible point. More specifically, the new feasible point can be constructed as 
\beq\label{opt_c_construct}
\breve \bc_{pk,j} \!=\! \bar \bc_{pk}\! +\! (-1)^j\sqrt{1\!-\!\|\bar \bc_{pk}\|^2}\disps \frac{\widetilde \bc_{pk}^{\perp}}{\|\widetilde \bc_{pk}^{\perp}\|},j\!=\!1,2,
\eeq
where $\widetilde \bc_{pk}^{\perp}$ is a non-zero orthogonal vector\footnote{It is always possible to find a non-zero orthogonal vector to the five vectors: $\bar \bc_{pk}$, $\widetilde \bc_0$ and $\bh_{l, pk}^i,l=1,2,3$, when the code length $M\geq 3$, i.e., the dimensional of the real-valued form code $\widetilde \bc_{pk}$ is greater than 6.} to $\bar \bc_{pk}$, $\widetilde \bc_0$ and $\bh_{l, pk}^i,l=1,2,3$. It is clear that $\breve \bc_{pk,j}$ has unit energy while adhering to the linear constraints in Problem \eqref{Pk_sca}. Moreover, it is always possible to find a $\breve \bc_{pk,j}$ such that it achieves a lower objective value compared to $\bar \bc_{pk}$, which is given by
\begin{align}\label{obj_diff}
	&(\bg_{pk}^i)^\trans\breve \bc_{pk,j}+g_{pk}^i-(\bg_{pk}^i)^\trans\bar \bc_{pk}-g_{pk}^i\nonumber\\
	&\quad=(-1)^j\sqrt{1-\|\bar \bc_{pk}\|^2}\disps \frac{(\bg_{pk}^i)^\trans\widetilde \bc_{pk}^{\perp}}{\|\widetilde \bc_{pk}^{\perp}\|}\nonumber\\
	&\quad\leq 0,
\end{align}
where the choice of $j$ depending on the sign of $(\bg_{pk}^i)^\trans\widetilde \bc_{pk}^{\perp}$.
\subsection{Proof of Proposition \ref{proposition:3}}\label{proof_propostion3}
	For the sake of argument, let $g_{pk}(\widetilde \bc_{pk};\widetilde \bc_{1k}^{i+1},\cdots,\widetilde \bc_{(p-1)k}^{i+1},\widetilde \bc_{pk}^{i},\cdots,\widetilde \bc_{Nk}^{i})$, $p=1,\cdots,N$ be the surrogate function associated with the $p$-th block at $(i+1)$-iteration of the block-MM algorithm. Following a similar line of reasoning as in \cite{razaviyayn2013unified,aubry2018new}, the feasibility of the sequence $\widetilde \bc^i$ to $\bar \calP_k$ is first established. This follows from the fact that an optimal solution fulfilling the energy equality to the approximate Problem \eqref{Pk_sca_relax} can be obtained, and such a solution is also feasible for $\bar{\calP}_k$, as ensured by the tight approximate constraints reported in \eqref{cAc_approx}-\eqref{hpk}.

Then, the following chain of inequalities hold true 
\begin{align}
	\widetilde f_k&(\widetilde\bc_{1k}^i,\cdots\!,\! \widetilde\bc_{Nk}^i) \nonumber\\
	&=g_{1k}(\widetilde \bc_{1k}^{i};\widetilde \bc_{1k}^{i},\cdots,\widetilde \bc_{pk}^{i},\cdots,\widetilde \bc_{Nk}^{i})\label{mon_eq1}\\
	&\geq \min_{\widetilde \bc_{1k}\in\widetilde \calS_{1k}^i} g_{1k}(\widetilde \bc_{1k};\widetilde \bc_{1k}^{i},\cdots,\widetilde \bc_{pk}^{i},\cdots,\widetilde \bc_{Nk}^{i})\label{mon_eq2}\\
	&=  g_{1k}(\widetilde \bc_{1k}^{i+1};\widetilde \bc_{1k}^{i},\cdots,\widetilde \bc_{pk}^{i},\cdots,\widetilde \bc_{Nk}^{i})\label{mon_eq3}\\
	&\geq \widetilde f_k(\widetilde\bc_{1k}^{i+1},\cdots\!,\! \widetilde\bc_{Nk}^i)\label{mon_eq4}\\
	&\vdots\label{mon_eq5}\\
	&\geq g_{Nk}(\widetilde \bc_{Nk}^{i};\widetilde \bc_{1k}^{i+1},\cdots,\widetilde \bc_{pk}^{i+1},\cdots,\widetilde \bc_{Nk}^{i})\label{mon_eq6}\\
	&\geq \min_{\widetilde \bc_{Nk}\in\widetilde \calS_{Nk}^i} g_{Nk}(\widetilde \bc_{Nk};\widetilde \bc_{1k}^{i+1},\cdots,\widetilde \bc_{pk}^{i+1},\cdots,\widetilde \bc_{Nk}^{i})\label{mon_eq7}\\
	&=  g_{Nk}(\widetilde \bc_{Nk}^{i+1};\widetilde \bc_{1k}^{i+1},\cdots,\widetilde \bc_{pk}^{i+1},\cdots,\widetilde \bc_{Nk}^{i})\label{mon_eq8}\\
	&\geq \widetilde f_k(\widetilde\bc_{1k}^{i+1},\cdots\!,\! \widetilde\bc_{Nk}^{i+1})\label{mon_eq9}
\end{align}
where $\widetilde \calS_{pk}^i=\{\widetilde \bc_{pk}\big|\|\widetilde\bc_{pk}\|^2 = 1,\widetilde \bc_0^\trans\widetilde\bc_{pk}\geq \zeta_1,(\bh_{l, pk}^i)^\trans\widetilde\bc_{pk}+  h_{l,pk}^i\geq\varepsilon,l=1,2,3\}$, the equality in \eqref{mon_eq1} holds because $g_{1k}(\widetilde \bc_{1k};\widetilde \bc_{1k}^{i},\cdots,\widetilde \bc_{Nk}^{i})$ is an upper bound of $\widetilde f_k(\widetilde\bc_{1k},\cdots\!,\! \widetilde\bc_{Nk}^i)$ and coincides with it at $\widetilde \bc_{1k}^{i}$; the inequality in \eqref{mon_eq2} follows from the fact that $\widetilde \bc_{1k}^{i} \in \widetilde \calS_{pk}^{i}$, by virtue of conditions \eqref{cAc_approx}-\eqref{hpk}; the equality in \eqref{mon_eq3} holds since $\widetilde \bc_{1k}^{i+1}$ is an optimal solution to both Problem \eqref{Pk_sca} and its approximated version \eqref{Pk_sca_relax}; the inequality in \eqref{mon_eq4} follows for the same reason as \eqref{mon_eq1}. By sequentially applying the above arguments to the remaining $N-1$ blocks, \eqref{mon_eq9} is obtained, which implies the claimed monotonicity, i.e.,
\beq
	\widetilde f_k^i=f_k(\widetilde\bc_{1k}^i,\cdots\!,\! \widetilde\bc_{Nk}^i)\geq \widetilde f_k(\widetilde\bc_{1k}^{i+1},\cdots\!,\! \widetilde\bc_{Nk}^{i+1})=\widetilde f_k^{i+1}.
\eeq
Finally, since the approximated objective function is bounded below by zero due to the non-negativity constraints $[\widehat \bS_{nk}(\widetilde\bc_{nk})]_{ll}\geq \varepsilon,l=1,2,3$, it follows that $\widetilde f_k^i$ converges to a finite value.

\bibliography{./lib/TaoWaveformRefs, ./lib/TaoWaveformForTrackingRefs}